\documentclass[11pt]{article}
\usepackage{amsmath,amsthm,amsfonts,amscd,eucal,latexsym,amssymb}
\usepackage{epsfig}  
\usepackage[latin1]{inputenc}
\usepackage{verbatim} 
\def\cB{{\ca B}}

\def\cD{{\ca D}}
\def\cE{{\ca E}}

\def\cP{{\ca P}}

\def\cS{{\ca S}}

\def\cW{{\ca W}}

\def\mD{{\mathcal D}}
\def\mE{{\mathcal E}}

\def\mW{{\mathcal W}}

\def\bC{{\mathbb C}}           

\def\bI{{\mathbb I}}

\def\bR{{\mathbb R}}

\def\bS{{\mathbb S}}

\def\gD{{\mathfrak D}}

\def\gF{{\mathfrak F}}
\def\gG{{\mathfrak G}}

\def\gS{{\mathfrak S}}

\def\gs{{\mathfrak s}}
\def\gw{{\mathfrak w}}
\def\gg{{\mathfrak g}}

\def\beq{\begin{eqnarray}}
\def\eeq{\end{eqnarray}}
\def\pa{\partial}
\def\at{\left(}               

\def\ct{\right)}              
\newcommand{\ca}[1]{{\cal #1}}         

\def\alp{\alpha}
\def\be{\beta}
\def\ga{\gamma}
\def\de{\delta}
\def\ep{\varepsilon}

\def\si{\sigma}
\def\om{\omega}

\def\th{\theta}

\def\Ga{\Gamma}
\def\De{\Delta}
\def\La{\Lambda}
\def\Si{\Sigma}
\def\Om{\Omega}

\newcommand{\wick}[1]{:\!{#1}\!:}
\newcommand{\supp}{\text{supp }}

\newtheorem{theorem}{Theorem}[section]
\newtheorem{proposition}{Proposition}[section]
\newtheorem{lemma}{Lemma}[section]
\newtheorem{definition}{Definition}[section]

\newcommand{\se}[1]{\section{#1}}

\def\vsp{\vspace{0.2cm}}
\def\vspp{\vspace{0.1cm}}

\def\sse #1 {\vsp\ifhmode{\par}\fi\refstepcounter{subsection}
  \noindent {\bf\thesubsection}. {\em #1}.\quad
  \addcontentsline{toc}{subsection}{\protect\numberline{\thesubsection} #1}%
  }

\def\ssb #1 {\vsp\ifhmode{\par}\fi\refstepcounter{subsection}
  \noindent {\bf\thesubsection.} {\bf #1.}\quad
  \addcontentsline{toc}{subsection}{\protect\numberline{\thesubsection} #1}%
  }

\def\ssa #1 {\ifhmode{\par}\fi\refstepcounter{subsection}
  \noindent {\bf\thesubsection.} {\bf #1.}\quad
  \addcontentsline{toc}{subsection}{\protect\numberline{\thesubsection} #1}%
  }

\def\remark #1 {\vsp\vspp\ifhmode{\par}\fi\noindent\noindent {\bf Remark.} {#1}\vsp\vspp\par}
\def\remarks #1 {\vsp\vspp\ifhmode{\par}\fi\noindent\noindent {\bf Remarks.} {#1}\vsp\vspp\par}

\begin{document}

\hfill{\sl September 2010}
\par
\bigskip
\par
\rm


\par
\bigskip
\LARGE
\noindent
{\begin{center}\bf Approximate KMS states for scalar and spinor fields\\in Friedmann-Robertson-Walker spacetimes\end{center}}
\bigskip
\bigskip
\par
\rm
\normalsize


\noindent {\bf Claudio Dappiaggi$^{1,a}$}, {\bf Thomas-Paul Hack$^{1,b}$}
and {\bf Nicola Pinamonti$^{2,c}$}\\
\par
\small
\noindent $^1$
II. Institut f\"ur theoretische Physik, Universit\"at Hamburg,
Luruper Chaussee 149,
D-22761 Hamburg, Germany.

\noindent $^2$
Dipartimento di Matematica, Universit\`a di Roma ``Tor Vergata'', Via della Ricerca Scientifica,
I-00133 Roma, Italy.\smallskip

\noindent $^a$  claudio.dappiaggi@desy.de, $^b$ thomas-paul.hack@desy.de, $^c$  pinamont@mat.uniroma2.it\\
 \normalsize

\par

\rm\normalsize
\noindent {\small Version of \today}

\rm\normalsize


\par
\bigskip

\noindent
\small
{\bf Abstract}.
We construct and discuss Hadamard states for both scalar and Dirac spinor fields in a large class of
spatially flat Friedmann-Robertson-Walker spacetimes characterised by an initial phase either of exponential or of power-law
expansion. The states we obtain can be interpreted as being in thermal equilibrium at the time when the scale factor $a$ has a specific value $a=a_0$. In the case $a_0=0$, these states fulfil a strict KMS condition on the boundary of the spacetime, which is either a cosmological horizon, or a Big Bang hypersurface. Furthermore, in the conformally invariant case, they are conformal KMS states on the full spacetime. However, they provide a natural notion of an approximate KMS state also in the remaining cases, especially for massive fields. On the technical side, our results are based on a bulk-to-boundary reconstruction technique already successfully applied in the
scalar case and here proven to be suitable also for spinor fields. The potential applications of the states we
find range over a broad spectrum, but they appear to be suited to discuss in particular thermal phenomena such as the cosmic neutrino background or the quantum state of dark matter.
\normalsize

\tableofcontents

\se{Introduction}
The mathematically rigorous formulation of quantum field theory on globally hyperbolic curved spacetimes has undoubtedly witnessed terrific progress  during the past years. Our understanding of the quantization scheme  within the algebraic framework initially allows us to specify which quantum states are physically relevant for any free field theory; based on these, one can then construct both Wick polynomials and interacting field theories on non-trivial backgrounds in a perturbative manner
\cite{BF00, HW01, HW02, BFV, HW03, BrFe09}.

Hence, one of the main challenges is to apply this mathematical scheme to concrete problems, and the first and foremost examples are certainly found in the realm of Cosmology. According to the standard modern paradigm, the Universe can be described by a homogeneous and isotropic solution of Einstein's equations, a Friedmann-Robertson-Walker (FRW) spacetime. This requirement entails that the metric is fully determined by a single dynamical quantity, the so-called scale factor $a(t)$, and a parameter which fixes the topology of the constant-time hypersurfaces to be spherical, flat, or hyperbolic. In view of recent measurements of the evolution of the Hubble function, we only consider the case of flat spatial sections -- so-called flat FRW spacetimes -- in the present work. Furthermore, we demand that the scale factor of the FRW metric displays an inflationary behaviour of either exponential or power-law type at early times. Our aim is to show that, under these circumstances, it is always possible to single out a so-called Hadamard state for both scalar and spinor quantum fields. The Hadamard condition is of particular relevance since it guarantees that the UV behaviour of the state mimics the one of the Minkowski vacuum, which in turn assures that the quantum fluctuations of observables such as the smeared components of the stress-energy tensor are bounded.

The scenario of a free scalar field has been extensively studied in the past and it has been found that, on a FRW spacetime with flat spatial sections, a Hadamard state invariant under the metric isometries can be singled out in various ways, see \cite{JuSc, Olbermann, Muri, Schlemmer} and \cite{DMP2, DMP3}. Here, we will focus on the approach first introduced in the last two last cited papers, where the notion of a distinguished asymptotic Hadamard ground state was discussed by means of a so-called bulk-to-boundary approach. This is a procedure which calls for identifying a preferred null, differentiable, codimension 1 submanifold of the full spacetime -- the boundary -- on which it is possible to construct a $*$-algebra of observables which contains the bulk one via an injective $*$-homomorphism. The merit of this construction is that the peculiar geometric structure of the auxiliary submanifold allows for identifying a boundary ground state whose pull-back in the bulk is both invariant under isometries and, more importantly, of Hadamard form.

The first aim of this work is to extend this procedure in order to obtain additional physically interesting states for free fields. As first realised in \cite{DMP4}, the structure of the boundary is such that it is also possible to naturally construct states for the boundary algebra of observables which are thermal, that is, they fulfil a KMS condition with respect to a suitable boundary translation. We will show that this idea can be implemented successfully in a cosmological framework and for scalar fields, yielding states which are both Hadamard in all cases and conformal KMS states in the case of massless fields conformally coupled to the scalar curvature. More interestingly, it turns out that, even though the KMS condition is not exactly fulfilled for massive theories, approximate thermodynamic relations always hold at early times. Hence, the states we find can be interpreted as being asymptotic thermal states in the early universe, where the departure from strict thermal equilibrium is controlled by the magnitude of the scale factor times the mass. We generalise these results even further by introducing Hadamard states which fail to satisfy a KMS condition on the boundary, {\it i.e.} at $a(t)=0$, but rather fulfil an approximate KMS condition at $a(t)=a_0>0$.

The second main goal of this paper is to prove that the bulk-to-boundary procedure can be extended to the case of free massive Dirac fields. We show that it is possible to construct a suitable boundary algebra of observables encompassing the bulk one also in this case, and that one can naturally identify preferred states on this boundary algebra whose pullback is always of Hadamard type and either a strict conformal KMS state for massless fields, or an approximate KMS state in the case of non-vanishing mass.

We want to stress that our analysis is not performed only out of a mere mathematical interest, but that there actually cogent physical reasons to study the topics we discuss in this work. On the one hand, one of the most interesting applications of the algebraic quantisation scheme to Cosmology is the analysis of solutions to the semiclassical Einstein equations for massive fields. Indeed, it has been found in \cite{DFP} that a scale factor driven by a massive scalar field in a Hadamard state naturally evolves to a late time de Sitter solution, thereby offering an explanation for dynamical dark energy. It is certainly mandatory to analyse this phenomenon also for fields of higher spin and to see if more observed features of the cosmological evolution can be modelled by the sole effects of quantum fields in more general Hadamard states. In fact, the very states constructed in this paper have already been analysed in this manner to confirm that the backreaction of free Dirac quantum fields on the spacetime curvature is qualitatively of the same nature as the one of free scalar fields, and it has been found that the energy density of dark matter, usually considered as an ad-hoc classical quantity, can be accurately modelled by free quantum fields in an approximate KMS state of the kind introduced in this work \cite{dhmp}.

On the other hand, an important object of study in modern cosmology is the cosmic neutrino background, which is usually considered to have a temperature of approximately 1.9 Kelvin (see for example \cite{Kolb}). However, since it is now well-known that neutrinos are massive and, hence, do not fulfil a conformally invariant equation of motion, it is impossible to associate to them an exact conformal KMS state. Therefore, the approximate KMS states for massive Dirac fields we introduce seem to be natural candidates to discuss the cosmic neutrino background accurately in terms of quantum field theory in curved spacetimes.

The paper is organised as follows: section 2 starts with a recollection of the geometrical properties of the spacetimes we consider and closes with a study of the properties of both a classical scalar and a classical Dirac field theory living thereon. It is particularly emphasised how to solve the Dirac equation in a flat Friedmann-Robertson-Walker spacetime in terms of a vector-valued, {\it diagonal}, second order hyperbolic partial differential operator. Section 3 instead is focused on the construction of the bulk-to-boundary correspondence. After recollecting how this mechanism works for scalar fields, we continue to fully develop it for spinors. Particularly, we prove the existence of an injective $*$-homomorphism from the bulk field algebra into the boundary counterpart. Section 4 is the core of the paper and it is here shown how to assign distinguished states for a free field theory on the boundary which are satisfying an exact KMS condition. Afterwards, the bulk counterpart is constructed, and it is shown that both in the scalar and in the Fermionic case the outcome is a Hadamard state which is in addition thermal when the classical dynamics of the bulk field is conformally invariant. In case conformal invariance is broken, we show that the result still fulfils approximate thermodynamic relations and that this even holds for states which are not exact KMS states on the boundary, but fulfil an approximate conformal KMS condition at finite times. Section 5 contains the conclusions, while the appendices are focused on reviewing the definition of Dirac fields on a four-dimensional globally hyperbolic spacetime and on proving the fall-off behaviour at infinity of classical Dirac modes in the considered flat FRW spacetimes.

\section{From the Geometry to the Classical Field Theory}\label{geom}
This section is bipartite: in the first part, encompassing the next two subsections, we recollect some facts already discussed in \cite{DMP2, DMP3, Nicola}. Particularly, we define the class of spacetimes we are interested in and recollect their geometric properties. Furthermore, we sketch the behaviour of a classical real scalar field living thereon. In the second part, we focus on classical Dirac spinors and provide a detailed account of their dynamics on cosmological spacetimes.

\subsection{Friedmann-Robertson-Walker Spacetimes with a Lightlike Cosmological Boundary}
In this paper, we shall consider spacetimes $M$ as being four-dimensional Hausdorff, connected,
smooth manifolds endowed with a Lorentzian smooth metric $g$ whose signature is $(-,+,+,+)$. Particularly,
we shall be interested in homogeneous and isotropic solutions of Einstein's equations with
{\em flat spatial sections}. Thus, the metric $g_{FRW}$ has the so-called Friedmann-Robertson-Walker form:
\beq\label{metric}
ds^2=-dt^2+a^2(t)\left[dr^2+r^2(d\theta^2+\sin^2\theta d\varphi^2\right],
\eeq
which is here written in spherical coordinates. The coordinate $t$, also known as {\bf cosmological time}, ranges {\it a
priori} over an open interval $I\subseteq\bR$ which we later constrain, whereas $a(t)\in C^\infty(I,\bR^+
)$, $\bR^+$ being the strictly positive real numbers. If one introduces the so-called {\bf conformal time} $\tau$
out of the defining differential relation $d\tau\doteq a^{-1}(t)dt$, the metric becomes:
\beq\label{metric2}
ds^2=a^2(\tau)\left[-d\tau^2+dr^2+r^2(d\theta^2+\sin^2\theta d\varphi^2)\right],
\eeq
which is manifestly a conformal rescaling of the Minkowski one.

This rather simple consideration initially prompted our attention towards this class of spacetimes, since the
flat background is the prototype of a large class of solutions of Einstein's equations with vanishing
cosmological constant, namely, the asymptotically flat spacetimes (see \cite{DMP} or \cite{Wald} for a recollection of the definition). From a geometrical point of view, these are
rather distinguished since, by means of a conformal compactification process, they can be endowed with a notion of
(conformal) past or future null boundary, usually indicated as $\Im^\pm$. This rather special structure was
successfully used in \cite{DMP} (see also references therein) as a tool to set up a bulk-to-boundary
correspondence ultimately yielding the possibility to construct a quasi-free invariant Hadamard state for a massless scalar field conformally
coupled to scalar curvature.

Therefore, it seemed natural to wonder if a similar bulk-to-boundary correspondence and related results could be obtained also for FRW spacetimes.
To this avail, the first question to answer is, under which conditions on the scale factor $a(t)$
past or future null-infinity can be genuinely associated to $(M,g_{FRW})$. Indeed, if we require that the functional form of $a(t)$ and its domain $I\ni t$ are such that the domain of $a(\tau)$ includes $(-\infty, -\tau_0)$ ($(\tau_0, \infty)$) for some $\tau_0<\infty$, then one can meaningfully assign to $(M,g_{FRW})$ conformal past (future) null infinity by means of the conformal equivalence with Minkowski spacetime. As we would somehow like to consider the field theoretical constructions on null infinity as ``initial conditions'' rather than ``final conditions'' for the field theories in the bulk spacetimes, we shall restrict attention to cases where $(M,g_{FRW})$ possesses a meaningful notion of past null infinity. However, all our constructions and results can be trivially extended to FRW spacetimes with future null infinity. Four simple examples of cosmological spacetimes with past null infinity are:
\begin{itemize}
 \item[a)] $a(t)=\left(\frac{\ga}{-t}\right)^n$, $n>0$, $\ga>0$, $t\in I=(-\infty,0)$ $\quad\Rightarrow\quad$ $a(\tau)=\left(\frac{\ga}{-\tau (n+1)}\right)^{\frac{n}{n+1}}$
\item[b)] $a(t)=\frac{t}{\ga}$, $\ga>0$,  $t\in I=(0,\infty)$ $\quad\Rightarrow\quad$ $a(\tau)=e^{\frac{\tau}{\ga}}$
\item[c)] $a(t)=\left(\frac{t}{\ga}\right)^n$, $n>1$, $\ga>0$,  $t\in I=(0,\infty)$ $\quad\Rightarrow\quad$ $a(\tau)=\left( \frac{\ga}{-\tau(n-1)}\right)^{\frac{n}{n-1}}$
\item[d)] $a(t)=e^{\frac{t}{\ga}}$, $\ga>0$,  $t\in I=(-\infty,\infty)$ $\quad\Rightarrow\quad$ $a(\tau)=\frac{\ga}{-\tau}$
\end{itemize}
Since only the behaviour of $a(t)$ at the lower bound of $I\ni t$ is essential, one can immediately obtain a large class of flat FRW spacetimes which posses past null infinity by requiring that the functional behaviour of $a(t)$ is asymptotically of one of the above mentioned kinds in the early past. The FRW spacetimes one obtains in this manner from the cases b), c), and d) are of special cosmological relevance, as they describe a cosmic history characterised by an early stage of either {\bf de Sitter inflation} -- d) -- or {\bf power-law inflation} -- b) and c). This is the physical motivation why we shall restrict to these cases in the following. The technical motivation is the fact that, as found in \cite{DMP2, Nicola}, in these spacetimes it is possible to establish a bulk-to-boundary relation for {\em massive} field theories, while, in FRW spacetimes which are asymptotically of class a), this is only possible for {\em massless} theories. To understand the reason behind this in simple terms, we briefly describe how past null infinity can be attached to $(M, g_{FRW})$ if the above described requirements are met.

Switching to the coordinates $U=\tan^{-1}(\tau+r)$ and $V=\tan^{-1}(\tau-r)$, the FRW metric reads
\beq\label{Einsteinform} ds^2=\frac{a^2(\tau(U,V))}{\cos^2U\cos^2 V}\left[-dUdV+\frac{\sin^2(U-V)}{4}(d\theta^2+\sin^2\theta d\varphi^2)\right],
\eeq where the factor in the square brackets is the line element proper to the metric of the {\bf Einstein static universe} $(M_E, g_E)$. As a result, flat FRW spacetimes can be conformally embedded into this spacetime, and the locus $U=-\frac\pi2$ in $M_E$ -- corresponding to the "boundary" $\tau+r=-\infty$ of a FRW spacetime with the domain of $a(\tau)$ containing $(-\infty, -\tau_0)$ -- can be meaningfully considered as a smooth hypersurface $\Im^-$ diffeomorphic to $\bR\times \bS^2$, {\it i.e.} past null infinity. A bulk-to-boundary construction as the one we seek to employ always requires to extend the definition of a field in the original spacetime -- here $(M,g_{FRW})$ -- to a field on the boundary, namely, $\Im^-$. How this can be achieved crucially depends on the behaviour of $g_{FRW}$ on past null infinity. In spacetimes which are asymptotically de Sitter, one finds that the prefactor of the square brackets in \eqref{Einsteinform} becomes constant on $\Im^-$. Hence, in this case, $(M,g_{FRW})$ can be extended beyond $\Im^-$, and the bad behaviour of $g_{FRW}$ on $\Im^-$ turns out to be rooted in an unlucky choice of coordinates. In this case, $\Im^-$ turns out to be a {\bf horizon} of the FRW spacetime. In contrast, in case a), the prefactor in \eqref{Einsteinform} diverges on $\Im^-$, while it vanishes in the remaining cases, {i.e.} in spacetimes with an early power-law inflation. Accordingly, $\Im^-$ describes the {\bf Big Bang} hypersurface in the power-law case. On the field theoretic side, the failure of $\frac{a^2(\tau(U,V))}{\cos^2U\cos^2 V}$ to be finite on $\Im^-$ is inherited by solutions of the classical field equations and forces us to introduce suitable conformal rescalings of the field in order to be able to map a theory in a finite manner from $(M,g_{FRW})$ to $\Im^-$. As a result, a field with mass $m$ becomes a field with mass $$\sqrt{\frac{a^2(\tau(U,V))}{\cos^2U\cos^2 V}}\,m\,.$$
Its is now clearly visible why a bulk-to-boundary construction for massive fields can only be established in spacetimes which are asymptotically of de Sitter or power-law inflationary type.

For the technical details of our analysis it will be convenient to discuss all constructions on the level of conformal time $\tau$. Hence, let us recall the related requirements on the scale factor $a(\tau)$ of the flat FRW spacetimes we wish to work with in the following, {\it viz.},
\beq\label{scalef}
a(\tau)=\frac{\gamma^{1+\de}}{(-\tau)^{1+\de}}+O\left(\frac{1}{(-\tau)^{2+\de+\epsilon}}\right),
\qquad\frac{da(\tau)}{d\tau}=-\frac{(1+\de)\gamma^{1+\de}}{(-\tau)^{2+\de}}+O\left(\frac{
1}{(-\tau)^{3+\de+\epsilon}}\right),\;\de\ge 0\,.
\eeq
Here, the formulae are meant to hold asymptotically for $\tau\to-\infty$. Furthermore, $\epsilon$ is a small positive constant that has to be added in order to obtain the desired regularity for the classical solutions near null infinity in the case $\de=0$. In favour of notational simplicity, we consider the special case $a(t)\sim t\ga^{-1}\Leftrightarrow a(\tau)\sim e^{\tau\ga^{-1}}$ to be represented by the limiting case $\de\to\infty$ in the above formulae.

As the interplay between bulk symmetries and the vector $\pa_\tau$ with the geometric structure of $\Im^-$ will play a pivotal role in the following, we shall now recall a few more details on the geometry of $\Im^-$. In the asymptotic power-law case, $\Im^-$ has manifestly the same properties as the past conformal boundary of an asymptotically flat spacetime, extensively discussed in \cite{DMP} and the references therein. In the asymptotic de Sitter case, the structure of $\Im^-$ is different due to the fact that it is a horizon, a full discussion can be found in \cite{DMP2}. In more detail, setting
$$\Om\doteq 2 \cos U = \frac{2}{\sqrt{1+(\tau-r)^2}}\,,$$
we can easily compute that $d\Om|_ {\Im^-}\neq 0$ and that the line element $\widetilde ds^2$ of the metric
$$\widetilde g\doteq \frac{\Om^2}{a^2}\,g_{FRW}$$ restricted to $\Im^-$ is of {\bf Bondi form}, namely,
$$\widetilde ds^2|_{\Im^-}=2d\Om dv+d\bS^2\,.$$
Here, $v\doteq\tau+r$ and $d\bS^2$ denotes the standard measure on the two-sphere $\bS^2$, $d\theta^2+\sin^2\theta d\varphi^2$. One can check that $\pa_\tau$ fulfils the conformal Killing equation
$$\mathcal{L}_{\partial_\tau}g_{FRW}=-2(\partial_\tau \ln(a))g_{FRW}$$ and is, hence, a {\bf conformal Killing vector} of the FRW spacetime. This vector turns out to be tangent to $\Im^-$ and its restriction to $\Im^-$ equals $\frac12 \pa_v$, whose integral curves in turn generate $\Im^-$. Altogether, we find that a conformal Killing vector $\pa_\tau$ of the bulk FRW spacetime becomes a proper Killing vector on its boundary $\Im^-$ endowed with the Bondi metric. However, the full symmetry group $G_{\Im^-}$ of $\Im^-$ is much larger than that. In the asymptotic power-law case, it is given by symmetry group of the conformal boundary of Minkowski spacetime, the so-called {\bf Bondi-Metzner-Sachs group} (see for instance \cite{DMP}). In the asymptotic de Sitter case, however, the situation is different since the FRW spacetime can be extended to and beyond $\Im^-$ {\em without} a conformal rescaling\footnote{Note that the conformal transformation relating $g_{FRW}$ and $\widetilde g$ which made $\widetilde g$ regular and of Bondi form on $\Im^-$ is proportional to the identity on $\Im^-$ in the asymptotic de Sitter case. Nevertheless, we have not omitted it in favour of a uniform notation.}. In fact, it turns that the symmetry group in this case is constituted by the set of diffeomorphisms of $\bR\times\bS^2$ which map the point $(v,\theta,\varphi)$ to $(e^{\alp(\theta,\varphi)}v+\be(\theta,\varphi), R(\theta,\varphi))$ with $\alp,\be\in C^\infty(\bS^2)$ and $R\in SO(3)$ \cite{DMP2}. A direct inspection of this formula reveals that the so-called {\bf horizon symmetry group} is given by the iterated semidirect product
\beq\label{group}
G_{\Im^-}=SO(3)\ltimes(C^\infty(\bS^2)\ltimes C^\infty(\bS^2)).
\eeq

\subsection{Classical Scalar Fields on FRW Spacetimes}\label{classicalscalar}

This subsection shall offer a brief introduction of both the symplectic space and the mode decomposition of a classical scalar field on a FRW spacetime with flat spatial sections. Although most, if not all, the material has already appeared elsewhere, we feel that it is worth devoting a few lines to this topic as a guideline for the discussion of the more complicated case of Dirac fields.

Hence, let us consider the field $\phi:M\to\bR$, whose dynamics are ruled by the conformally coupled {\bf Klein-Gordon equation}
\beq\label{eqmotionscalar}
P \phi =0 \;, \qquad P = -\Box + \frac{1}{6} R + m^2\;,
\eeq
where $R$ is the scalar curvature, while $m^2\geq 0$ is the squared mass. Since all FRW spacetimes as in \eqref{metric} are globally hyperbolic, one can associate to \eqref{eqmotionscalar} a meaningful Cauchy problem. Particularly, each smooth and compactly supported initial datum yields a solution $\phi\in C^\infty(M)$. The linear set of such solutions can be characterised as $\cS(M)=\{\phi_f\;|\;P\phi_f=0\;\textrm{and}\;\exists f\in C^\infty_0(M),\;\phi_f=Ef\}$. Here,  $E=E^--E^+:C^\infty_0(M)\to C^\infty(M)$ stands for the causal propagator, the difference between the advanced and the retarded fundamental solution. Furthermore, for each $\phi_f\in\cS(M)$, $\supp\phi_f\subseteq J^+(\supp(f))\cup J^-(\supp(f))$, and $\cS(M)$ forms a symplectic space if endowed with the following weakly non-degenerate {\bf symplectic form},
$$
\si_M(\phi_f,\phi_g)\doteq \int\limits_{\bR^3} \at {\phi_f}\frac{d\phi_g}{d\tau} -  \frac{d\phi_f}{d\tau} \phi_g  \ct a(\tau)^2d^3x=E(f,g).
$$
Notice that $\si_M$ is independent of $\tau$ since it is in general independent of the Cauchy surface constituting the integration domain.

Although the above characterisation of the space of classical solutions of \eqref{eqmotionscalar} is, apart from the special but inessential choice of Cauchy surface in the definition of $\si_M$, valid on every globally hyperbolic spacetime, we can provide a more detailed characterisation of $\cS(M)$ in the present case. To wit, by making use of the big isometry group of $(M,g_{FRW})$, we can specify each element of $\cS(M)$ via an explicit decomposition in Fourier modes. Particularly, the following two propositions have been proven in \cite{DMP2, DMP3}:

\begin{proposition}\label{KGmodes}
For every $k\in [0,\infty)$, let $T_k(\tau)$ be a complex solution of the differential equation
\beq\label{TK}
(\pa^2_\tau +k^2+a^2m^2)T_{k} = 0\,,
\eeq
which fulfils the following normalisation condition
$$
T_k\pa_\tau\overline{T_k} - \overline{T_k}\pa_\tau T_k \equiv i\;.
$$
Suppose, furthermore,  that the functions $k\mapsto T_k(\tau)$ and
$k\mapsto \pa_\tau T_k(\tau)$
are both polynomially bounded  for large $k$ uniformly in $\tau$ and in $L^2([0,\underline{k}],kdk)$ for every $\underline{k}>0$.
Under these hypotheses, every element $\phi$ of $\cS(M)$ can be decomposed in modes as
\beq\label{bosondecomposition}
\phi(\tau,\vec{x})= \int\limits_{\bR^3}
\at \phi_{\vec{k}}(\tau,\vec{x}) \widetilde{\phi}(\vec{k})      + \overline{\phi_{\vec{k}}(\tau,\vec{x}) \widetilde{\phi}(\vec{k})}
\ct
d^3k,\eeq
where $\phi_{\vec{k}}(\tau,\vec{x})$ is solution of the conformally coupled Klein-Gordon equation of the form
\beq\label{bosonmodes}
\phi_{\vec{k}}(\tau,\vec{x}) = \frac{ T_{k}(\tau)   e^{i\vec{k}\vec{x}}}{(2\pi)^{3/2}a(\tau)}.
\eeq
\end{proposition}

\noindent As shown in \cite{DMP3, Nicola} modes $T_k$ fulfilling the assumptions of the above proposition can be concretely constructed in the class of spacetimes we consider via a converging perturbation series; moreover, they can be chosen in such a way that the following initial conditions are satisfied
$$
\lim_{\tau\to-\infty} e^{ik\tau}T_k(\tau) = \frac{1}{\sqrt{2k}} \;, \qquad
\lim_{\tau\to-\infty} e^{ik\tau}\pa_\tau T_k(\tau) = -i\sqrt{\frac{k}{2}} \;.
$$

\noindent The second proposition deals instead with the invertibility of the decomposition \eqref{bosondecomposition}.

\begin{proposition}
Under the hypotheses of the preceding proposition, for every $\phi\in\cS(M)$, \eqref{bosondecomposition} can be inverted as
$$
\widetilde{\phi}(\vec{k})=i\si_M\left(\overline{\phi_{\vec{k}}},\phi\right)\,.
$$
Furthermore, $\widetilde{\phi}(\vec{k})$ is square integrable and, as $|\vec{k}|$ diverges, it decays faster then any inverse power of $|\vec{k}|$.
\end{proposition}

\noindent Notice that, since the symplectic form does not depend on $\tau$ and since $\phi_{\vec{k}}$ is a solution of \eqref{eqmotionscalar}, $\widetilde{\phi}(\vec{k})$ is time-independent as well. This observation and the preceding discussion entail
$$
\si_M(\phi_1,\phi_2) = -i\int\limits_{\bR^3}
\left(\, \overline{\widetilde{\phi}_1(\vec{k})} \widetilde{\phi}_2(\vec{k})
- {\widetilde{\phi}_1(\vec{k})} \overline{\widetilde{\phi}_2(\vec{k})} \,\right)
d^3k,
$$
which is once more manifestly time-independent.

To conclude this section and as a later useful tool, we recall that also the causal propagator assumes a somehow more manageable form when decomposed in modes, namely \cite{LuRo},
$$
E(\tau_x,\vec{x},\tau_y,\vec{y}) =
-i
\int\limits_{\bR^3}
\overline{\phi_{\vec{k}}(\tau_x,\vec{x})}
\phi_{\vec{k}}(\tau_y,\vec{y})-
\phi_{\vec{k}}(\tau_x,\vec{x})
\overline{\phi_{\vec{k}}(\tau_y,\vec{y})}\;d^3 k\;.
$$

\subsection{Classical Dirac Fields on FRW Spacetimes}\label{classicalpart}

\noindent We shall now focus on the analysis of the classical dynamics of Dirac spinors living on the class of FRW spacetimes we consider. The full discussion of the well-posedness of this dynamical system requires the introduction of several additional structures and, hence, a rather lengthy detour. To avoid such a detour, we point the interested reader either to \cite{DHP}, where an extensive review is available, or to the appendix \ref{appdirac} which contains a small resum\'e of the employed notations and conventions.

As described in the above-mentioned appendix, a Dirac spinor $\psi$ (respectively cospinor $\psi^\prime$) in a given spacetime $M$ is a smooth global section of the Dirac bundle $DM$ (respectively dual Dirac bundle $DM'$). As discussed in lemma 2.1 of \cite{DHP}, $DM$ is well-defined and trivial whenever $M$ is globally hyperbolic and simply connected, thus, on a FRW spacetime in particular. In other words, in $(M,g_{FRW})$, the standard picture of (co)spinors being vector-valued functions  $\psi:M\to\bC^4$ and $\psi^\prime:M\to\bC^4$ is valid. That said, we call {\em dynamically allowed} any $\psi$ which fulfils the so-called {\bf Dirac equation}
$$D\psi=(-\displaystyle{\not}\nabla+m\bI_4)\psi=0,$$
where $\bI_4$ stands for the four dimensional identity matrix, whereas $\displaystyle{\not}\nabla\doteq\gamma^a e_a^\mu\nabla_\mu$. Here $\ga^a$, $a\in\{0,1,2,3\}$, are $4\times 4$ complex matrices which form a representation of the Clifford algebra $Cl(1,3)$ and which are usually referred to as $\ga-$matrices. Due to the employed signature, the explicit form of the $\ga$-matrices slightly differs from the standard one; if we denote by $\si_i$ the standard Pauli matrices, we can set
\beq\label{repgamma}
\ga_0\doteq  i
\at\; \begin{matrix} \bI_2 & 0\\ 0 &-\bI_2 \end{matrix}
\ct\;
, \qquad \ga_j\doteq i \begin{pmatrix} 0 & \si_j  \\ -\si_j & 0  \end{pmatrix}\  \;.\qquad j\in\{1,2,3\}
\eeq

\noindent Henceforth, lower-case Roman letters mean that the associated quantities are expressed in terms of a Lorentz frame $e_a$, $a\in\{0,1,2,3\}$. This is a set of four global sections of the tangent bundle which fulfil the relation $g(e_a,e_b)=\eta_{ab}$, where $\eta$ stands for the flat Minkowski metric. All lower-case Roman indices are thus raised and lowered via $\eta_{ab}$. Furthermore, $\nabla$ denotes the covariant derivative on (tensor products of) the full Dirac bundle (see definition 2.10 and lemma 2.2 in \cite{DHP}), its explicit form thus encompasses the spin connection coefficients.

\remark{As a notational convention, whenever a $\ga$-matrix appears with a definite index, this refers to its expression in the non-holonomic basis, {\it e.g.} $\ga^0\doteq \ga^a|_{a=0}$. If we choose the conformal coordinates displayed in \eqref{metric2}, then the $\ga$-matrices $\ga^\mu$ are related to the former via a multiplicative factor $a(\tau)^{-1}$.}

Let us now recall that, according to theorem 2.3 in \cite{DHP}, the Dirac operator $D$ admits unique advanced and retarded fundamental solutions $S^\pm$ which are continuous linear maps from $\mD(DM)$, the set of compactly supported smooth sections of the Dirac bundle $DM$ into $\mE(DM)$, the space of smooth sections. Furthermore, it holds for all $f\in\mD(DM)$ that
$$
DS^\pm=\text{id}_{\mD(DM)}=S^\pm D\,,\qquad \supp (S^\pm f)\subseteq J^\pm(\supp (f)).
$$
Hence, we can introduce the {\bf causal propagator} $S\doteq S^--S^+$ associated to $D$, which in turn allows us to characterise the space of smooth solutions of the Dirac equation with compactly supported initial data as
\beq\label{spaceofsol}
\gS(M)\doteq\left\{\psi_f\;|\; D\psi_f=0\;\textrm{and}\;\exists f \in\mD(DM)\;\textrm{such that}\;\psi_f=Sf\right\}.
\eeq
The set $\gS(M)$ carries a natural Hermitian structure, namely, for all $f_1,f_2\in\mD(DM)$ we can define an inner product as 
$$
\mathfrak{s}_M(\psi_{f_1},\psi_{f_2})\doteq
i\int\limits_{\Sigma}\psi_1^\dagger\displaystyle{\not}n\,\psi_2\,d\mu(\Sigma)
=
-i\int\limits_M f_1^\dagger S f_2\;d\mu(M),
$$
where $\Sigma$ is any Cauchy surface of the underlying background whose unit normal vector is denoted by $n$. Such choice plays no role since independence of the integral from $\Sigma$ was already proven in proposition 2.2 of \cite{Dimock}. Here, the superscript $\dagger$ refers to the {\bf Dirac conjugation map} which, for all $\psi\in\mE(DM)$ is defined as $\psi^\dagger\doteq\psi^*\beta$ where $*$ denotes adjoint with respect to the standard inner product on $\bC^4$, whereas $\beta$ is the Dirac conjugation matrix, that is, the unique element of $SL(4,\bC)$ which fulfils
$$
\beta^*=\beta\,,\quad \gamma^*_a=-\beta\gamma_a\beta^{-1}\,,
$$
and $\quad -i\beta n^\mu \ga_\mu$ is a positive definite matrix for all timelike and future-pointing vector fields $n$.
In the representation \eqref{repgamma} we have chosen, $\beta=-i\ga_0$.

\remark{We only concentrate on Dirac spinors since the behaviour of cospinors can be inferred from the one of spinors by applying the Dirac conjugation map to all relevant equations.}

We now specialise our treatment to any FRW spacetime $(M, g_{FRW)}$ with flat spatial sections whose scale factor $a(\tau)$ is of the form \eqref{scalef}. To analyse the behaviour of a Dirac field living thereon, we first discuss the choice of a Lorentz frame, and two natural possibilities exist. On the one hand, since the spacetime $(M, g_{FRW})$ can be conformally embedded\footnote{In the asymptotically de Sitter case, the FRW spacetime at hand can even be isometrically embedded in a larger space which, however, differs from the Einstein static universe $(M_E,g_E)$.} in a larger spacetime $(M_E,g_E)$ which contains the horizon $\Im^-$, we could simply choose the frame $e_a$ as the global one on $M_E$. The advantage of this choice would be that this Lorentz frame would be automatically well defined on $\Im^-$. On the other hand, in the following discussion it will be crucial to exploit the conformal flatness of $g_{FRW}$ displayed in \eqref{metric2}. Hence, we shall choose the standard Minkowski frame, say $\tilde{e}_a$, and then define the one on $(M, g_{FRW})$ as
$$
e_a = a(\tau)^{-1} \tilde{e}_a.
$$
Notice that, due to the divergent factor $a$, the frame $e_a$ cannot neither be extended to $\Im^-$ nor to $(M_E,g_E)$.
Although we will be able to circumvent this problem implicitly, we refer the reader interested in further explicit details to \cite{Thomas}. In the chosen conformally flat Lorentz frame, the covariant derivative on $DM$, and, consequently, the Dirac equation assume a somehow more manageable form, {\it viz.}
\beq\label{Diraceq}
{a^{-\frac{5}{2}}}\left(-\gamma^0\partial_\tau-\gamma^i\nabla_i+am\right)a^{\frac{3}{2}}\;\psi=0\,,
\eeq
where $\nabla_i$ is the Cartesian gradient along the spatial directions.

As in the scalar case, we would like to have a more explicit Fourier mode characterisation of solutions of the Dirac equation on $(M,g_{FRW})$ with compactly supported smooth initial data. To this avail, and in view of the high symmetries of a FRW spacetime, it would be desirable to reduce \eqref{Diraceq} to a diagonal form.  The natural approach which calls for diagonalising $D'D$ is somehow not well-suited to our aims, since, after performing the Fourier transform along the spatial directions, the resulting ODE displays a term which is both complex and linear in $k$, a situation which leads to potential practical difficulties when constructing a state. A different and ultimately more effective diagonalisation procedure has been introduced in \cite{Barut} and we shall now discuss an improved and more clear version of this approach. To start, let us introduce the modified operator

\beq\label{defmod}
\gD\doteq -\ga^0 D=
a^{-\frac{5}{2}}\left( \begin{array}{cc}
\pa_\tau+iam & -\vec{\nabla}\cdot\vec{\sigma} \\
-\vec{\nabla}\cdot\vec{\sigma} & \pa_\tau-iam
\end{array} \right)a^{\frac{3}{2}},
\eeq

\noindent where $\si_i$ are the Pauli matrices and $\cdot$ denotes the inner product in $\bR^3$. The equations $D\psi=0$ and $\gD\psi=0$ are equivalent, as $-\ga^0$ is invertible. If we define $\gD^\prime\doteq-aD^\prime\ga^0$, then we can compute
\begin{gather}\label{defsecondorder}
\gD\gD^\prime= a^{-\frac{5}{2}}\left( \begin{array}{cc}
(-\pa^2_\tau+\De-ia^\prime m+a^2m^2)\bI_2 & 0\\
0 & (-\pa^2_\tau+\De+ia^\prime m+a^2m^2)\bI_2 \end{array} \right)
a^{\frac{3}{2}}\,,
\end{gather}
where $\prime$ denotes a derivative with respect to $\tau$ and $\Delta$ is the Laplace operator on $\bR^3$. Notice that, on Minkowski spacetime, the above procedure is equivalent to considering $DD^\prime$, which is already diagonal in this special case. The following lemma summarises the above discussion and will play a key role in the study of the fall-off behaviour at infinity of solutions of the Dirac equation:

\begin{lemma}\label{Cauchy}
The following Cauchy problems for a spinor $\psi:I\times\bR^3\to\bC^4$ are equivalent:
$$
\begin{array}{lll}
1)\quad\left\{
\begin{array}{l}
D\psi=0\\
\psi|_{t=0}=f\in C^\infty_0(\bR^3)
\end{array}
\right.,& & 2)\quad \left\{
\begin{array}{l}
DD'u=(-\nabla_\mu\nabla^\mu+\frac{R}{4}+m^2)u=0\\
u|_{t=0}=0\\
(\partial_t u)|_{t=0}=\gamma_0 f
\end{array}
\right.,\\
& & \\
3)\quad\left\{
\begin{array}{l}
\gD\psi=0\\
\psi|_{t=0}=f\in C^\infty_0(\bR^3)
\end{array}
\right.,& & 4)\quad \left\{
\begin{array}{l}
P_\gD\widetilde u\doteq (\gD\gD')\widetilde u=0\\
\widetilde u|_{t=0}=0\\
(\partial_t \widetilde u)|_{t=0}=-\frac{f}{a(0)}
\end{array}
\right.,
\end{array}$$
where $\psi=D'u$ and $\psi =\gD'\widetilde u$.
\end{lemma}
\begin{proof}$1)\Longleftrightarrow 2)$ has been already proven in theorem 2.2 of \cite{DHP}. $1)\Longleftrightarrow 3)$ is immediate once we recall that $\gD$ differs from $D$ by a multiplicative pre-factor which is an element of $SL(2,C^\infty(\bR,(0,\infty)))$. $3)\Longleftrightarrow 4)$ can be proven as follows: suppose $4)$ holds true, and let us introduce $\psi=\gD'\widetilde u$, then $\psi(0)=\gD'\widetilde u(0)=a(0)\gamma_0 D\widetilde u|_\Sigma$. The restriction of $\widetilde u$ and of its derivatives to $\Sigma$ is meaningful since $P_\gD$ is, per direct inspection of \eqref{defsecondorder}, a set of four second order hyperbolic differential operators. Thus, $\widetilde u$ is a unique smooth solution on the whole $M$. In order to achieve compatibility with the initial condition of $3)$, one sets $\psi(0)=f$ and thus $a(0)\gamma_0 \gamma^\mu\nabla_\mu\widetilde u|_\Sigma=f$. Since $\nabla_\mu\widetilde u|_\Sigma=-n_\mu\frac{\partial\widetilde u}{\partial n}$ where $n$ is the normal vector field to the Cauchy surface, then, owing to $n\equiv \pa_t$, one obtains per direct substitution $-a(0)\gamma_0 \gamma^0\frac{\partial\widetilde u}{\partial t}|_\Sigma=f$, which, thanks to the assigned initial condition for $4)$ and to the identity $\gamma_0\gamma^0= 1$ yields
 the sought result. Notice that we have just proved that $4)$ implies $3)$, but uniqueness of the solution of the Cauchy problem also entails the converse. The transitivity property of the proven equivalences suffices to infer the statement of the lemma.
 \end{proof}

\noindent The main advantage of this lemma is that, in order to construct a solution of the Dirac equation in a cosmological spacetime, we can just focus on the fourth Cauchy problem. Furthermore, using the fact that the operator $\gD\gD'$ is diagonal, we can discuss the four components of the solutions $\widetilde{u}$ independently. It turns out that a convenient mode basis of solutions of $\gD\gD'$ is given by
\beq\label{buildblock}
p_{\vec{k},l}(\tau,\vec{x})\doteq
\frac{u_{k,l}(\tau)
e^{i\vec{k}\cdot\vec{x}}}
{(2\pi a)^{\frac{3}{2}}},
\eeq
where $k\doteq |\vec{k}|$ and
$$
u_{k,1}\doteq\left(\begin{array}{c}\chi_{k,1}\\0\\0\\0\end{array}\right), \quad u_{k,2}\doteq\left(\begin{array}{c}0\\\chi_{k,1}\\0\\0\end{array}\right),\quad u_{k,3}\doteq\left(\begin{array}{c}0\\0\\\overline{\chi_{k,1}}\\0\end{array}\right),\quad u_{k,4}\doteq\left(\begin{array}{c}0\\0\\0\\\overline{\chi_{k,1}}\end{array}\right)
$$
$$u_{k,5}\doteq\left(\begin{array}{c}\chi_{k,2}\\0\\0\\0\end{array}\right), \quad u_{k,6}\doteq\left(\begin{array}{c}0\\\chi_{k,2}\\0\\0\end{array}\right),\quad u_{k,7}\doteq\left(\begin{array}{c}0\\0\\\overline{\chi_{k,2}}\\0\end{array}\right),\quad u_{k,8}\doteq\left(\begin{array}{c}0\\0\\0\\\overline{\chi_{k,2}}\end{array}\right).$$
Here, $\chi_{k,1}$ and $\chi_{k,2}$ constitute two linearly independent solutions of the ODE
\beq\label{ODE}
\cP\chi_{k,j}\doteq \left(\pa^2_\tau+k^2+a^2m^2-ia^\prime m\right)\chi_{k,j}(\tau)=0\,. \quad j\in\{1,2\}
\eeq

\remark{Although $\cP$ is not a real differential operator and, hence, $\chi_{k,1}$ and $\chi_{k,2}$ cannot be related by complex conjugation, we are still free to choose $\overline{\chi_{k,1}}$ and $\overline{\chi_{k,2}}$ as a basis of the space of solutions of $\overline{\cP}$.}

As in the study of the solutions of the Klein-Gordon equation conformally coupled to scalar curvature, we would like to use \eqref{buildblock} to construct a mode decomposition of each solution $\psi$ of the Dirac equation. Yet, in view of lemma \ref{Cauchy}, which guarantees that a solution of the Dirac equation can be found from one of $\gD\gD'$ by application of $\gD'$, we can expect that not all 8 basis modes provided by \eqref{buildblock} are needed to construct a mode expansion of $\psi$. As a matter of fact, the following result implicitly entails that already the first four modes listed in \eqref{buildblock} are complete.

\begin{proposition}\label{Diracmod}
Let $\chi_{k}$ be a smooth solution of \eqref{ODE} which satisfies the following normalisation condition
\beq\label{normalisation}
|(\pa_\tau+iam)\chi_{k}|^2+k^2|\chi_{k}|^2\equiv1.
\eeq
Suppose, furthermore, that for any fixed $\tau$ and $k_1>0$, both the functions $k\mapsto\chi_k(\tau)$ and $k\mapsto\pa_\tau\chi_k(\tau)$ are in $L^2{([0,k_1], k^2 dk)}$ and they grow at most polynomially in $k$.
Then, for any solution $\psi\in\gS(M)$, it holds
\beq\label{diracexpansion}
\psi(\tau,\vec{x})=\sum\limits_{l=1}^4\int\limits_{\bR^3}\widetilde\psi_l(\vec{k}) \psi_{\vec{k},l}(\tau,\vec{x})\,d^3k\,,
\eeq
where
\beq\label{diracmodes}
\qquad
\psi_{\vec{k},l}(\tau,\vec{x})\doteq
\gD^\prime
\frac{u_{k,l}(\tau)
e^{i\vec{k}\cdot\vec{x}}}
{(2\pi a)^{\frac{3}{2}}}.
\eeq
\end{proposition}

\noindent If one takes into account \eqref{normalisation}, it is also possible to invert the above decomposition as

\begin{proposition}\label{diracprop}
Under the hypotheses of the preceding proposition, the mode decomposition of any Dirac field $\psi\in\gS(M)$ can be inverted as
\beq\label{diracmodeinversion}
\widetilde{\psi_l}(\vec{k}) = -i\int\limits_{\bR^3}  a^3(\tau) \,
\psi^\dagger_{\vec{k},l}(\tau,\vec{x}) \ga_0 \, \psi(\tau,\vec{x}) \,d^3x
= \gs_M(\psi_{\vec{k},l},\psi)\, .
\eeq
Furthermore, each $\widetilde{\psi_l}(\vec{k})$ is a square integrable function which decays faster then any inverse power of $k$.
\end{proposition}

\begin{proof}
In order to show that either \eqref{diracmodeinversion} inserted in \eqref{diracexpansion} or vice versa yields an identity is a matter of a long, tedious, but direct computation. However, the statement on the regularity of $\widetilde\psi_l(\vec{k})$ requires a closer look. It descends both from the hypotheses on $\chi_k$ formulated in the previous proposition and from the fact that, at fixed $\tau$, $\psi(\tau,\vec{x})$ is an $\vec{x}$-dependant four-vector whose components are compactly supported smooth functions. Hence, they are square integrable and their Fourier transform is rapidly decreasing in $k$. If one combines this result with the fact that $\widetilde\psi_l(\vec{k})$ is a linear combination of the product of the Fourier transform of compactly supported smooth functions together with either $ima \chi_k(\tau)$, $(\sigma\cdot \vec{k})\chi_k(\tau)$ or $\pa_\tau \chi_k(\tau)$, then the sought result follows.
\end{proof}

\noindent Notice that, as in the case of a scalar field, the $\tau-$independence of $\widetilde{\psi_l}(\vec{k})$ stems from both the Cauchy surface independence of $\gs_M$ and the fact that $\psi_{\vec{k},l}$ is a solution of the Dirac equation.

\remark{
For the class of spacetimes under investigation, a set of solutions of \eqref{ODE} $\chi_k$, which fulfil the properties required in the above two propositions, are constructed in lemma \ref{modesregularity}. They are uniquely determined by the initial conditions we shall assume to be fulfilled in the following, {\it viz.} 
$$
\lim_{\tau\to-\infty} e^{ik \tau }\chi_{k}(\tau) = \frac{1}{\sqrt{2}k}\; ,\qquad
\lim_{\tau\to-\infty} e^{ik \tau } \pa_\tau \chi_{k}(\tau) = -\frac{i}{\sqrt{2}} \;.
$$
}

\noindent  As a first profitable consequence of the mode decomposition we can obtain a convenient expression for the conserved Dirac Hermitian product on FRW spacetimes with flat spatial sections.

\begin{lemma}\label{usefullemma} For every $\psi_1$ and $\psi_2$ in $\gS(M)$, it holds
\beq\label{hermitianproduct}
\gs_M (\psi_1,\psi_2 ) = \sum_{l=1}^4\;\int\limits_{\bR^3}\, \overline{\widetilde\psi_{l,1}(\vec{k})}\; \widetilde\psi_{l,2}(\vec{k})
\;d^3k,
\eeq
where
$\widetilde\psi_{l,j}(\vec{k}) \doteq \gs_M(\psi_{\vec{k},l}, \psi_j) $ for $j\in \{1,2\}$.
\end{lemma}

\begin{proof}
Let us rewrite
$$
\gs_M (\psi_1,\psi_2 ) =
-i\int\limits_{\Si}\psi_1^\dagger\displaystyle{\not}n\,\psi_2\,d\mu(\Si)= \int\limits_{\bR^3} \psi_1^*\psi_2\; a^3(\tau)\;d^3x
$$
where we have used both $\psi^\dagger=\psi^*\beta$ and the freedom to choose $\Si$ in such a way that $n\equiv\pa_t$ and 
$-i\beta \displaystyle{\not}{n}=\bI_4$.
We now insert the expansion
$$
\psi_2(\tau_x,\vec{x}) = \int\limits_{\bR^3} d\vec{k} \sum_{l=1}^4  \gs_M(\psi_{\vec{k},l}, \psi_2)\;\psi_{\vec{k},l}(\tau_x,\vec{x})
$$
into $\gs_M (\psi_1,\psi_2 )$. 
Since the $x$-integration is over a compact set while the integrand is rapidly decreasing for large $k$, we can use the theorem of  dominated convergence to switch the order of the $k$- and $x$-integration. Afterwards, by using the sesquilinearity of $\gs_M$ and by noticing that
$
\gs_M(\psi_1,\psi_{\vec{k},l} ) = \overline{\widetilde\psi_{l,1}(\vec{k})}\;,
$
the sought result follows.
\end{proof}

To conclude the section, we remark that the following mode decomposition of the causal propagator can be obtained: Provided that the modes $\chi_k$ fulfil the conditions imposed in the two main propositions of this section, one finds \cite{Thomas}
$$
S(\tau_x,\vec{x},\tau_y,\vec{y})= i \sum_{l=1}^4 \int\limits_{\bR^3}   \psi^{\phantom{\dagger}}_{\vec{k},l}(\tau_x,\vec{x}) \,\psi^\dagger_{\vec{k},l}(\tau_y,\vec{y}) \;d^3k \;.
$$
In the last section of this paper we will use this expression to derive a mode expansion for the two-point function of the states we shall introduce.

\section{Quantum Field Theory in the Bulk and on the Cosmological Boundary}

The aim of this section is to use the constructions of the previous section in order to achieve two main results. On the one hand, we will show that it is possible to set up a genuine Fermionic quantum field theory on past null infinity $\Im^-$. On the other hand, we shall prove that one can encode the information of a bulk quantum field theory into the counterpart on $\Im^-$.

As in the previous section, we shall first discuss the case of a massive real scalar field conformally coupled to scalar curvature in order to recapitulate the main features of the construction we aim for in a scenario where it is already known to work. Subsequently, we prove that the same procedure can be successfully carried out in the case of a Dirac field as well.

\subsection{The Bulk-to-Boundary Correspondence for Real Scalar Fields}

This subsection recollects some of the results already proven in \cite{DMP2, DMP3, Nicola} concerning the projection of real scalar fields to $\Im^-$. The key step is the following: while the classical theory on $(M,g_{FRW})$ naturally contains a symplectic space of solutions of the equation of motion, there is no dynamical content on the cosmological horizon. Therefore, we are forced to perform a choice of a symplectic space on the boundary which can be only justified {\it a posteriori}. To wit, as in the above-mentioned references, we introduce
$$
\cS(\Im^-)=\left\{ \Phi\in C^\infty(\Im^-) \;|\;\Phi\;\textrm{and}\;\pa_v \Phi  \in L^2(\Im^-;dv d\bS^2(\theta,\varphi)) \right\}\;,
$$
where $d\bS^2(\theta,\varphi)$ is the standard measure on the $2$-sphere. This is a strongly non degenerate symplectic space if endowed with the symplectic form
$$
\si_{\Im^-}(\Phi_1,\Phi_2) = -\int\limits_{\bR\times\bS^2}\at \Phi_1\pa_v \Phi_2 -  \Phi_2\pa_v \Phi_1 \ct  dv d\bS^2 \;.\qquad\forall\; \Phi_1,\Phi_2\in\cS(\Im^-)
$$
The motivation to choose $(\cS(\Im^-),\si_{\Im^-})$ is the existence of a symplectomorphism
$$
\Ga_S:\cS(M)\to \cS(\Im^-)
$$
whose explicit action on the elements of $\cS(M)$ is defined as
\begin{gather*}
\Ga_S(\phi)(v,\theta,\varphi)\doteq
\lim_{u\to -\infty} - u\, a\, \phi(\tau(u,v), \vec{x}(u,v,\theta,\varphi)).
\end{gather*}
Here, we have first implicitly switched from the Cartesian coordinates $\vec{x}=(x,y,z)$ on the bulk Cauchy surface isomorphic to $\bR^3$ to the spherical ones $(r=|\vec{x}|,\theta,\varphi)$ and introduced the null coordinates $v=\tau+r$ and $u=\tau-r$ afterwards.
The above limit was already computed in the proof of theorem 4.4 in \cite{DMP2} for the asymptotic de Sitter case, but the same computation can be trivially repeated for the asymptotic power-law case, as it relies on the fall-off behaviour of the classical modes towards $\Im^-$ which is the better the larger $\de$ in \eqref{scalef} is. The result found in \cite{DMP2} is
\begin{gather*}
\Ga_S(\phi)(v,\theta,\varphi)=
\frac{1}{\sqrt{2\pi}}
\int_0^\infty \at e^{-ikv} \,  \widetilde{\phi}(k,\pi-\theta,\pi+\varphi) + e^{+ikv} \,   \overline{\widetilde{\phi}(k,\pi-\theta,\pi+\varphi)} \ct \sqrt{\frac{k}{2}}\, dk\;.
\end{gather*}

\remark{The {\it ansatz} for the map $\Ga_S$ is based on the following preliminary considerations. In Minkowski spacetime, smooth solutions $\phi$ of the massless Klein-Gordon equation are known to decay as $(-u)^{-1}$ towards $\Im^-$, that is, in the limit $u\to-\infty$. To wit, due to the fact that Minkowski spacetime can be conformally embedded into the Einstein static universe and that the scalar field has conformal weight $1$, $$\widetilde \phi\doteq \frac{2}{\sqrt{(1+u^2)(1+v^2)}}\,\phi$$
\noindent is a solution of the massless, conformally coupled Klein-Gordon equation in $M_E$ \cite{DMP}. Hence, $\widetilde \phi$ is manifestly finite on $\Im^-$ and the decay behaviour of $\phi$ follows. The effect of the projection map $\Ga_S$ can now be understood as follows. The multiplication with the conformal factor $a$ transforms the massive, conformally coupled field $\phi$ on $(M,g_{FRW})$ into a scalar field on Minkowski spacetime with mass $am$. This mass term is, however, not essential for the fall-off behaviour of the associated solutions, as it is finite or vanishing on $\Im^-$. Hence, the additional factor of $-u$ in $\Ga_S$ is both necessary and sufficient to obtain a finite projection of solutions in $(M,g_{FRW})$ to $\Im^-$.}

\vsp

Since $\widetilde{\phi}$ decays faster than any inverse power of $k$ and since it is square integrable,
$\Ga_S(\phi)$ is an element of $\cS(\Im^-)$. Furthermore, as anticipated above, the following proposition (see theorem 4.2 in \cite{DMP2})
holds
\begin{proposition}\label{conservscal}
The symplectic form $\si_M$ is conserved under the projection $\Ga_S$, that is, for every $\phi_1$ and $\phi_2$ in $\cS(M)$, it holds
$$
\si_M(\phi_1,\phi_2) = \si_{\Im^-}(\Ga_S \phi_1,\Ga_S \phi_2)\;.
$$
Thus, the map $\Ga$ is an injective symplectomorphism from  $\cS(M)$ into $\cS(\Im^-)$.
\end{proposition}

\noindent This proposition relating classical field theories entails a close relationship between the bulk and the boundary algebra of quantum observables. As a matter of fact, let us recall that to every real vector space $S$ endowed with a weakly non-degenerate symplectic form $\sigma_S$ one can associate a unique (up to isometric $*$-isomorphisms) Weyl $C^*$-algebra whose generators $W(s)$, $s\in S$, fulfil the defining relations:
$$ a)\;W(s)^*=W(-s)\qquad b)\;W(s)W(s')=e^{\frac{i}{2}\sigma_S(s,s')}W(s+s')\,,\quad\forall s,s'\in S\,.$$
In the scenario at hand, it is thus clear that it is possible to associate a Weyl $C^*$-algebra both to the boundary -- $\cW(\Im^-)$ -- and to the bulk, say $\cW(M)$. Furthermore, in view of the above proposition it holds (see theorem 4.2  in \cite{DMP2}) that

\begin{proposition}\label{scalarbtb}
The symplectomorphism $\Ga_S$ induces an injective $*-$homomorphism of $C^*-$algebras $\imath_S:\cW(M)\to\cW(\Im^-)$, fully determined by
$$
\imath_S W(\phi)\doteq W(\Ga_S \phi)\,, \qquad \forall \phi \in \cS(M)\;.
$$
Furthermore, for every state, that is, for every continuous, positive, normalised, linear functional $\omega:\cW(\Im^-)\to\mathbb{C}$, there is  a counterpart $\omega_M:\cW(M)\to\mathbb{C}$ unambiguously defined as
$$
\omega_M\left(W(\phi)\right)\doteq\omega\left(\imath_S(W(\phi)\right)=\omega\left(W(\Ga_S\phi)\right)\,.
$$
\end{proposition}

\subsection{The Bulk-to-Boundary Correspondence for Dirac Fields}

In this subsection we will show how the bulk-to-boundary correspondence introduced for the scalar field can be adapted to Dirac fields. To this avail, the first tool we need is the counterpart of the Dirac bundle on the boundary. Although $\Im^-$ could be considered as a genuine manifold on its own, the question whether it admits a spin structure can not be answered automatically with the results we have invoked up to now, as $\Im^-$ endowed with the Bondi metric is not a globally hyperbolic spacetime. However, if we recall that $\Im^-$ is defined as a codimension $1$ embedded submanifold of the globally hyperbolic Einstein static universe $(M_E,g_E)$, we can exploit definition 5.3 in \cite{Husemoller} and introduce $D\Im^-\doteq\iota^*(DM_E)$, where $\iota:\Im^-\hookrightarrow M_E$ is the natural embedding map. Furthermore, on account of corollary 6.7 in \cite{Husemoller} and of the isomorphism\footnote{Note that $DM_E$ is trivial since $(M_E,g_E)$ is both globally hyperbolic and simply connected.} $DM_E\simeq M_E\times\mathbb{C}^4$, we can conclude that $D\Im^- \simeq\Im^- \times \bC^4$; consequently, the associated smooth sections can be equivalently interpreted as elements of $C^\infty(\Im^-,\bC^4)$.

As in the previous subsection, we have to specify a suitable boundary configuration space whose choice is justified only {\it a posteriori}. To wit, we call
\beq\label{Diracspace}
\gS(\Im^-) \doteq \{ \Psi \in C^\infty(\Im^-,\bC^4)  \cap L^2(D\Im^-;dv\,d\bS^2)\;, \;  \Xi(\vec{a})\Psi(v,\vec{n})=\Psi(v,\vec{n}) \},
\eeq
where $L^2(D\Im^-;dv\,d\bS^2)$ is isometric to $L^2(\Im^-;dv\,d\bS^2)\otimes\bC^4$. Furthermore $\vec{n}\in\bR^3$ has unit Euclidean-norm and thus corresponds to a point on the two-sphere, whereas $\Xi(\vec{n})$ is the following matrix-operator
\beq\label{Xi}
\Xi(\vec{n})\doteq
\begin{pmatrix}
 0 &  \vec{n}\cdot \vec{\si}\\
\vec{n}\cdot \vec{\si} &  \, 0
\end{pmatrix},
\eeq
where $\vec{n}\cdot\si$ is the Euclidean inner product of $\vec{n}$ and the Pauli vector. Notice that, in sharp contrast to the scalar case, a reminiscence of the bulk equation of motion prevails on the boundary via the constraint $\Xi(\vec{n})\Psi(v,\vec{n})=\Psi(v,\vec{n})$.  This identity implies that only two of the four components of an element of $\gS(\Im^-)$ are independent.

$\gS(\Im^-)$ can be naturally endowed with the Hermitian form
$$
\gs_{\Im^-}(\Psi_1,\Psi_2) \doteq    {i} \int\limits_{\bR\times \bS^2}   \Psi_1^\dagger\ga^0 \Psi_2\; dv\,d\bS^2,
$$
where, again, $\Psi^\dagger= -i\Psi^* \ga_0 $ while $*$ is the standard adjoint map on $\bC^4$.
This Hermitian form is non-degenerate on $\gS(\Im^-)$ because $(\vec{n}\cdot \vec{\si})^2 =\bI_2$ and, hence, 
$$
\gs_{\Im^-}(\Psi_1,\Psi_2)  =  2  \int\limits_{\bR\times \bS^2} \sum_{l\in\{1,2\}}   (\Psi_1^*)_l  (\Psi_2)_l\; dv\,d\bS^2,
$$
where only the two independent components of $\Psi_1^*$ and $\Psi_2$ have been summed. Furthermore, since $\Im^-$ is topologically equivalent to  $\bR\times\bS^2$, we are free to Fourier-transform each $\Psi\in\gS(\Im^-)$ as
$$
\widetilde{\Psi}_l(k,\theta,\varphi) \doteq \frac{1}{\sqrt{2\pi}} \int
\limits_{-\infty}^\infty dv\; e^{ikv}\, \Psi_l(v,\theta,\varphi)\;.
 \;\qquad l\in \{ 1,2,3,4\}
$$
Accordingly, the Hermitian product reads
$$
\gs_{\Im^-}(\Psi_1,\Psi_2) =
\int\limits_{\bR\times\bS^2} \sum_{l=1}^4
\overline{\widetilde\Psi_{l,1}(k,\theta,\varphi)}
\widetilde\Psi_{l,2}(k,\theta,\varphi)\, dk\,d\bS^2\,.
$$

The main result of this subsection is the following proposition which justifies the choice for $\gS(\Im^-)$. Regarding coordinates we employ the same notation as in the previous subsection.
\begin{proposition}\label{horizonprojection}
The map $\Ga_D:\gS(M) \to \gS(\Im^-)$
$$
\Ga_D(\psi)\doteq \lim_{u \to -\infty}    - u\; a^{3/2}\; \psi(\tau(u,v),\vec{x}(u,v,\theta,\varphi))
$$
possess the following properties:
\begin{itemize}
\item[a)] $\Ga_D$ is a well-defined, that is, $\Ga_D(\gS(M))\subseteq\gS(\Im^-)$,
\item[b)] If we set $\Psi\doteq\Ga_D{(\psi)}$  and denote the four-vector representing its Fourier transform by $\widetilde\Psi$, it holds
\beq\label{boundaryfield}
\widetilde\Psi(k,\vec{n})  = |k|\;
(\bI_4+\Xi(\vec{n}) ) \; \at
\Theta(k) \;
\begin{pmatrix}
\widetilde\psi_1 \at k,  -\vec{n}\ct \\
 \widetilde\psi_2 \at k,  -\vec{n}\ct  \\
0\\
0
\end{pmatrix}
+
\Theta (-k) \;
\begin{pmatrix}
0\\
0\\
 \widetilde\psi_3 \at -k,  \vec{n} \ct \\
 \widetilde\psi_4 \at -k,  \vec{n} \ct
 \end{pmatrix}\;
 \ct,
\eeq
where $\vec{n}$ is still a unit vector in $\bR^3$, $\Xi$ is the matrix defined in \eqref{Xi}, and
$\widetilde{\psi_l}(\vec{k})$ are the mode expansion coefficients of \eqref{diracexpansion},
\item[c)] $\Ga_D$ preserves the Hermitian forms, that is,
$$
\gs_M(\psi_1,\psi_2) = \gs_{\Im^-}(\Ga_D(\psi_1), \Ga_D(\psi_2))\,.
$$
\end{itemize}
\end{proposition}

\begin{proof}
To prove a), we follow once more the proof of theorem 4.4 in \cite{DMP2}. We start from the identity\begin{gather*}
\lim_{u\to-\infty}(-u) a^{3/2} \psi(u,v,\theta,\varphi)
=
\lim_{u\to-\infty}\frac{-u}{(2\pi)^{3/2}}\int\limits_{\bR^3}
\sum_{l=1}^4\;\hat\gD' u_{k,l}(\tau) e^{i\vec{k}\vec{x}}    \widetilde\psi_l(\vec{k})\;d^3k\;,
\end{gather*}
where $\hat\gD'=a^{\frac{3}{2}}  \gD' a^{-\frac{3}{2}}$. If we write $\vec{k}$ in spherical coordinates $k,\hat\theta,\hat\varphi$ where $\hat\theta$ measures the angle between $\vec{x}$ and $\vec{k}$ and if we restore the radial coordinate so that $v-u=2r$, we obtain
\begin{gather*}
\Psi(v,\theta,\varphi)=
\lim_{u\to-\infty}\frac{-u}{(2\pi)^{3/2}}\int\limits_{\bR^+\times\bS^2} \hat\gD' u_{k,l}(\tau) e^{ik {r} \cos (\hat\theta)}    \widetilde\psi_l(\vec{k})\; k^2\sin(\hat\theta) \; dk\, d\hat{\theta}\, d\hat{\varphi}\;.
\end{gather*}

On account of the properties and estimates of $\chi_k$ discussed in \ref{classicalpart} and appendix \ref{modeanalysis} and, setting  $c=\cos(\hat\theta)$, we obtain
\begin{gather*}
\Psi(v,\theta,\varphi)=
\lim_{u\to-\infty}\frac{-u}{(2\pi)^{3/2}}\int\limits_{0}^{2\pi} \int\limits_{-1}^1\int\limits_{0}^\infty
\sum_{s=\{+,-\}}
\at \widetilde{D}'_{s}(\vec{k})
e^{s\, (-ik\tau)} e^{ik {r} c}  \;
+O(|u|^{-\epsilon})
\ct
\widetilde\psi_s(\vec{k})
  \; \frac{k}{2} \; dk\, dc\, d\hat{\varphi},\;
\end{gather*}
where $\widetilde{\psi}_+\doteq(\widetilde{\psi}_1,\widetilde{\psi}_2,0,0)$,
$\widetilde{\psi}_-\doteq(0,0,\widetilde{\psi}_3,\widetilde{\psi}_4)$ and
$\widetilde{\psi} = \widetilde{\psi}_++\widetilde{\psi}_-$.
In the preceding expression $\widetilde{D}'_{\pm}(\vec{k})$ is the $4\times 4$ matrix
\beq\label{matricesboundary}
\widetilde{D}'_{\pm}(\vec{k})
\doteq  - i
\begin{pmatrix}
 \pm\, k\, \bI_2 & -\vec{k}\cdot \vec{\si}\\
-\vec{k}\cdot \vec{\si} &  \pm\,  k\, \bI_2
\end{pmatrix},
\eeq
while the $O(|u|^{-\epsilon})$ contribution descends from the estimates in appendix \ref{modeanalysis} and, in the asymptotic de Sitter case, from an expansion of $\chi^0_{k}(\tau)$ valid for large $\tau$
that can be derived from equation 8.421(9) in \cite{Grad}. In the asymptotic power-law case, the $O(|u|^{-\epsilon})$ contribution is identically vanishing.\\
\indent The multiplicative $u$ factor in front of the limits can be cancelled via an integration by parts in $c$. Afterwards, also the term proportional to $O(|u|^{-\epsilon})$ vanishes in the limit of divergent $u$ thanks to dominated convergence.
What is left after these operations consists of three terms, each to be evaluated as $u\to -\infty$:
\begin{gather*}
\frac{-u}{2r(u,v)} \frac{i}{(2\pi)^{3/2}}
\int\limits_{0}^{2\pi} \int\limits_{0}^\infty
\sum_{s\in\{+,-\}}
\widetilde{D}'_{s}(-\vec{k})
 e^{s(-ik\tau)}
 e^{-ik {r} }  \widetilde\psi_s(k,\pi,\hat\varphi)
   \; dk\, d\hat{\varphi}+\;
\\
-
\frac{-u}{2r(u,v)} \frac{i}{(2\pi)^{3/2}}
\int\limits_{0}^{2\pi} \int\limits_{0}^\infty
\sum_{s\in\{+,-\}}
\widetilde{D}'_{s}(\vec{k})
e^{s(-ik\tau)} e^{ik {r} }  \widetilde\psi_s(k,0,\hat\varphi)
   \; dk\, d\hat{\varphi}+\;
\\
+
\frac{-u}{2r(u,v)} \frac{i}{(2\pi)^{3/2}}
\int\limits_{0}^{2\pi} \int\limits_{-1}^1\int\limits_{0}^\infty
\sum_{s\in\{+,-\}}
\widetilde{D}'_{s}(\vec{k})
e^{s(-ik\tau)} e^{ik {r} c}
 \,\pa_c  \widetilde\psi_s(k,c,\hat\varphi) \; dk\, dc\, d\hat{\varphi}\;.
\end{gather*}
The second and the third term can be seen to vanish in the large $u$-limit by an application of the Riemann-Lebesgue theorem since the ratio $-u/(2r)$ tends to $1$  and $\pa_c \widetilde\psi_l$ is a regular integrable function. Notice that, in the last term, there is a potential obstruction in the application of the Riemann-Lebesgue theorem due to the behaviour of the integrand at $c=-1$. Yet, this problem can be overcome if one splits the domain of integration in $[-1,-1+\epsilon]$ and $[-1+\epsilon,1]$. The second contribution vanishes for large $-u$ and for all $\ep>0$, while the first one yields a finite result whose value is regulated by $\epsilon$ itself and, hence, vanishing as $\epsilon\to 0$.
We are left with the first of the three integrals and, since both $\widetilde\psi_s(k,0,\hat\varphi)$ and $\widetilde\psi_s(k,\pi,\hat\varphi)$ do not depend on $\hat\varphi$, this variable can be integrated out, yielding
\begin{gather*}
\Psi(v,\theta,\varphi)=
\frac{i}{\sqrt{2\pi}}
\int\limits_{0}^\infty
dk\;
\widetilde{D}'_{+}(-\vec{k})
 e^{-ikv}
 \widetilde\psi_+(-\vec{k})
 -
\widetilde{D}'_{-}(\vec{k})
 e^{ikv}
 \widetilde\psi_-(\vec{k}),
\end{gather*}
where $\vec{k}=(k,\theta,\varphi)$. Due to the form of $\widetilde{D}'_\pm$, the previous expression can be rewritten as
\begin{gather*}
\Psi(v,\vec{n})=
(\bI_4+\Xi(\vec{n}))\frac{1}{\sqrt{2\pi}}
\int\limits_{0}^\infty
dk\;k\;
\at
 e^{-ikv}\;
 \widetilde\psi_+ \at k,  -\vec{n}\ct
 +
 e^{ikv} \;
\widetilde\psi_-\at k,  \vec{n}\ct
\ct,
\end{gather*}
where $\Xi(\vec{n})$ is defined in \eqref{Xi}. Thanks to the regularity of $\widetilde{\psi}$ stated in proposition \ref{diracprop} and regardless of the multiplication by $k$, every component of $\Psi$ is a smooth square-integrable function.
This allows to conclude the proof of a) by noticing that the constraint in the definition of $\gS(\Im^-)$ is fulfilled by $\Psi$ since $\Xi(\vec{n})(\bI_4+\Xi(\vec{n}))=\bI_4+\Xi(\vec{n})$. Statement b) holds because the Fourier transform of the previous expression becomes equal to \eqref{boundaryfield} if we change the variable $k$ in $-k$ in the second contribution to the integrand and introduce the appropriate Heaviside step functions.\\
\indent In order to prove c), we insert the Fourier decomposition of $\Ga_D(\psi_1)$ and $\Ga_D(\psi_2)$ found in b) in the definition of $\gs_{\Im^-}$ to obtain
$$
\gs_{\Im^-}(\Ga_D{(\psi_1)},\Ga_D{(\psi_2)}) =   \int\limits_{\bR^+\times \bS^2} \sum_{l=1}^4 \overline{\widetilde\psi_{l,1}(k,\theta,\varphi) } \widetilde\psi_{l,2}(k,\theta,\varphi)    k^2 dk\,d\bS^2\,.
$$
The right hand side of the above equality coincides with $\gs_M(\psi_1,\psi_2)$ evaluated in Fourier space and in spherical coordinates.
\end{proof}

\remark{In this case, the {\it ansatz} for $\Ga_D$ is motivated as follows. As the conformal weight of a Dirac field is $\frac{3}{2}$, the factor $a^\frac32$ in $\Ga_D$ transforms the massive Dirac field in $(M,g_{FRW})$ into a Dirac field in Minkowski spacetime with mass $am$. The fall-off behaviour of corresponding solutions with compactly supported initial data is again independent of this time-dependent mass term. Since the massless Dirac field in Minkowski spacetime can be understood as a collection of four massless scalar fields, the factor $-u$ in $\Ga_D$ is, as already discussed in the context of $\Ga_S$, necessary and sufficient to obtain a well-defined projection of solutions to $\Im^-$.}

The last proposition is the analogue of proposition \ref{conservscal} for the Dirac fields, namely, the key building block to establish a relation between the bulk and the boundary algebra of quantum observables. Hence, let us start introducing these algebras following the discussion presented in \cite{Araki70} and further developed in \cite{Sanders, DHP, Sanders3, Thomas}. The essential idea is that, in order to describe the field algebra of Dirac fields $\cB(M)$, it is convenient to treat spinor and cospinor fields as a single and combined object. We refer to \cite{Sanders, DHP} for an in-depth analysis of the construction and analysis of this double-algebra and only summarise the key steps in the following. To wit, by $\cB(M)$ we indicate the field $*-$algebra generated by the identity and the linear functionals $B(f)$ with
$f=f_1\oplus f_2\in\cD(DM\oplus DM^*)$ subject to 
$$B(Df\oplus D^\prime h)=0,\qquad B^*(f)=B(\Gamma f)\,,$$
where $\Gamma f\doteq(f_2^\dagger,f_1^\dagger)$ with $\dagger$ denoting the Dirac conjugation.  Furthermore, all generators are required to satisfy the anticommutation relations
$$
\{ B(f),B(h)\}=iS^\oplus(f,h)
$$
where  $S^\oplus(f_1\oplus f_2,h_1\oplus h_2)=iS(h_2,f_1)+iS(f_2,h_1)\doteq \gs_M(\psi_{f_1},\psi_{h_2})+\gs_M(\psi_{f_2},\psi_{h_1})$
and where $S$ is the causal propagator of the Dirac equation. In this framework, we can recover spinors and cospinors as $\psi^\dagger(f_1)\doteq B(f)$ with $f=f_1\oplus 0$ and $\psi(f_2)=B(f)$ with $f=0\oplus f_2$.

The assignment of a field algebra to the theory on the horizon is slightly more involved. Let us start noticing that, if we define $\Psi^\dagger$ on $\Im^-$ as $-i \Psi^*\ga_0$ for all $\Psi\in\gS(\Im^-)$, we can mimic the construction of $\cB(M)$ to construct its boundary counterpart $\cB(\Im^-)$. To wit, following \cite{Thomas}, we define $\gS^\oplus(\Im^-)\doteq\gS(\Im^-)\oplus \gS(\Im^-)^\dagger$ where $\gS(\Im^-)^\dagger$ is the set formed by $\Psi^\dagger$ for every $\Psi$ in $\gS(\Im^-)$. $\gS^\oplus(\Im^-)$ can be equipped with the inner product
$$
\gs_{\Im^-}^\oplus\left((a_1,a_2^\dagger),(b_1,b_2^\dagger)\right) \doteq  \gs_{\Im^-}(b_2,a_1)+\gs_{\Im^-}(a_2,b_1)
$$
and the conjugation $j$
$$
j(a_1,a_2^\dagger)\doteq (a_2,a_1^\dagger)\;.
$$
We summarise the above discussion in the following definition.
\begin{definition}
The boundary algebra of Dirac fields $\cB(\Im^-)$ is the topological $*-$algebra generated by the linear functionals $f\mapsto B(f)$ with $f\in \gS^\oplus(\Im^-)$ together with the following conditions.
\begin{itemize}
\item[a)] The fields fulfil the canonical anticommutation relations (CAR)
$$
\{B(f),B(g)\} = \gs_{\Im^-}^\oplus(f,g)\quad\forall f,g\in\gS^\oplus(\Im^-)\,.
$$
\item[b)] The $*$-operation is specified by the antilinear involution $*:\cB(\Im^-)\to\cB(\Im^-)$ defined as 
$$
B(f)^*=B(j(f))\;\quad\forall f\in\gS^\oplus(\Im^-)\,.
$$
\item[c)] $\cB(\Im^-)$ is endowed with the quotient topology which descends from the local Fr\'echet topology of $\bigoplus_n (C^\infty(\Im^-,\bC^4))^{\otimes n}$.
\end{itemize}
\end{definition}

\noindent On the horizon, there is no equation of motion to be taken into account, but a reminiscence of the Dirac equation has been implemented in the construction of $\gS^\oplus(\Im^-)$.

We are now ready to formulate the main proposition of this subsection which establishes a relation between the boundary and the bulk algebra. The proof is a straightforward application of the previous definitions and of proposition \ref{horizonprojection}.

\begin{proposition}\label{Dirachom}
The map $\imath_D:\cB(M)\to \cB(\Im^-)$ unambiguously determined by its action on the fields $B(f)$, $f\in\cD(DM\oplus D^*M)$ as 
$$
\imath_D(B(f_1\oplus f_2^\dagger))\doteq B(\Ga_D(\psi_{f_1})\oplus\Ga_D(\psi_{f_2})^\dagger)\quad\forall f=f_1\oplus f_2\in\cD(DM\oplus D^*M)
$$
is an injective $*-$homomorphism. Furthermore, for every state, that is, for every continuous, positive, normalised, linear functional $\omega:\cB(\Im^-)\to\mathbb{C}$, there is a counterpart $\omega_M:\cB(M)\to\mathbb{C}$ unambiguously defined as
$$
\omega_M(B(f))\doteq\omega(\imath_D(B(f_1\oplus f_2))\quad\forall f=f_1\oplus f_2\in\cD(DM)\oplus\cD(D^*M)\,.
$$
\end{proposition}

\section{Boundary States and their Hadamard Bulk Counterpart}

In the previous section, we have proven that it is possible to induce a state for a scalar and Dirac quantum field theory on the bulk spacetime $(M,g_{FRW})$ by a state defined on the relevant quantum theory living on past null infinity. However, this construction of a distinguished state would be moot if the result could not be shown to be physically meaningful. Particularly, it is mandatory to check if the constructed state, say $\omega_M$, fulfils the so-called Hadamard condition. From a physical point of view, this condition guarantees that the UV behaviour of $\omega_M$ mimics the one of the Poincar\'e-invariant Minkowski vacuum and, hence, that the quantum fluctuations of observables, such as the smeared components of the stress-energy tensor, are bounded. From a formal point of view, the requirement for $\omega_M$ to be Hadamard boils down to a constraint on the form of the wave front set of the integral kernel proper to the truncated two-point function associated of the state $\omega_M$ \cite{Radzikowski, Radzikowski2, Kratzert, Hollands, SahlmannVerch, Sanders, Sanders2}. Since a cohesive and fully comprehensible explanation of all the mathematical tools underlying these notions would require a lengthy appendix, we deem it more appropriate to point an interested reader to the just mentioned references.

Our goal is to follow the same procedure employed in \cite{DMP,DMP2,DMP3, Nicola}: we will show, that, mostly thanks to the huge symmetry displayed by past null infinity $\Im^-$, it is possible to assign a distinguished state $\omega_{\Im^-}$ to the field algebra on $\Im^-$ whose pull-back to the bulk algebra is of Hadamard form. Furthermore, we will show that, both in the scalar and in the Dirac case, $\omega_{\Im^-}$ can be 
chosen to satisfy an exact KMS condition on the horizon, a property which will turn out to be preserved in the bulk under suitable circumstances. All states we construct are automatically quasi-free, that is, their structure is completely determined by their two-point functions. However, as recently proven in \cite{Sanders2}, this requirement is irrelevant for the discussion of the Hadamard condition.

In the case of a scalar field, the assignment of a two-point function $\om_2$ is tantamount to providing a distribution on $\mD(M\times M)$. However, in discussing (charge-conjugation invariant) Hadamard states for Dirac fields, we are concerned with two non-vanishing distributions $\gw^+(f,h)\doteq\om(\psi(h)\psi^\dagger(f))$ and $\gw^-(f,h)\doteq \om(\psi^\dagger(f)\psi(h))$, where $f\in\mD(DM)$ and $h\in\mD(D^*M)$. In other words, both $\gw^+$ and $\gw^-$ lie in $\mD'(DM\boxtimes D^*M)$ and they are related by the anticommutation relations
$$
\gw^-(f,h)+\gw^+(f,h)= i S(h,f)
$$
A quasifree state on the boundary algebra $\cB(\Im^-)$ is analogously determined by the choice of two distributions which are related by the boundary CAR.

\subsection{The Scalar Field Case}

This case is certainly not a novel one as the quest to construct a bulk Hadamard state by means of a boundary one was already pursued in \cite{DMP2, DMP3, Nicola}. Yet, in these papers, the possibility of defining a state on $\Im^-$  which fulfils a KMS condition was neither mentioned, nor analysed. Hence, we shall first briefly recollect the construction in the case of vanishing temperature $T=\beta^{-1}=0$ case and afterwards extend it to $\beta<\infty$.

To wit, we adopt here the definition of a quasi-free state introduced in \cite{Kay} and consider $\omega:W(\Im^-)\to\mathbb{C}$ unambiguously determined by
\beq\label{groundstate}
\omega(W(\Phi))=e^{-\frac{\mu(\Phi,\Phi)}{2}},\quad\forall\Phi\in\mathcal{S}(\Im^-)
\eeq
where
$$\mu(\Phi,\Phi')\doteq\int\limits_{\bR\times\bS^2}2k\Theta(k)\overline{\widehat{\Phi}(k,\theta,\varphi)}\widehat{\Phi}'(k,\theta,\varphi)\;dk\,d\bS^2\,,$$
and $\Theta(k)$ denotes the Heaviside step function. Furthermore, $\widehat\Phi$ stands for the Fourier-Plancherel transform (see appendix C in \cite{Moretti06}) defined for all $\Phi\in\cS(\Im^-)$ as
$$
\widehat\Phi(k,\theta,\varphi) = \frac{1}{\sqrt{2\pi}} \int\limits_\bR e^{ik v}\; \Phi({v,\theta,\varphi})\; dv \,.
$$
Notice that $\mu$ satisfies the constraint that every quasifree state (see \cite{Kay}, \cite[app. A]{Moretti06}) has to fulfil, namely
\beq\label{const}
|\sigma_{\Im^-}(\Phi,\Phi')|\leq 4|\mu(\Phi,\Phi)|^{\frac{1}{2}}|\mu(\Phi',\Phi')|^{\frac{1}{2}}\,.
\eeq
One of the most striking properties of $\omega$ descend from the study of its interplay with the boundary group of isometries $G_{\Im^-}$, namely, the BMS group in the asymptotic power-law case, or the horizon symmetry group \eqref{group} in the asymptotic de Sitter case. In fact, it turns out that $\om$ is invariant under the $*$-automorphisms $\alpha_\gg:\mW(\Im^-)\to\mW(\Im^-)$ on the boundary Weyl algebra induced by every element $\gg\in G_{\Im^-}$ via $\alpha_\gg(W(\Phi))\doteq W(\Phi\circ \gg^{-1})$ for all $\Phi\in\mathcal{S}(\Im^-)$. Particularly, summarising the content of theorem 4.1 in \cite{DMP2}, it turns out that $\omega$ is the unique pure and quasifree state on $\mW(\Im^-)$ which is invariant under the automorphic action of the one-parameter subgroup of $G_{\Im^-}$ generated via the exponential map by $\partial_v$, the generator of the rigid translations on $\Im^-$ along the $\bR$-direction. Moreover, one finds that $\om$ is a ground state with respect to the $v$-translations.

However, as found in \cite{DMP4}, the simple construction of the boundary ground state on the level of Fourier modes allows to define a boundary KMS state with respect to the translations along $v$ by means of the inner product
\beq\label{KMSboundary}
\mu_\beta(\Phi_1,\Phi_2)= \int\limits_{-\infty}^\infty \frac{2k}{1-e^{-\beta k}}\;   \overline{{\widehat\Phi_1}(k,\theta,\varphi)} \widehat\Phi_2(k,\theta,\varphi)   \;  dk\,d\bS^2
\eeq
for any $\beta>0$. In more detail, the following statements can be proven.

\begin{proposition}
The two-point function $\mu_\beta$ defined in \eqref{KMSboundary} induces a quasifree state $\om_\beta:\cW(\Im^-)\to\bC$ via $\om_\beta(W(\Phi))\doteq e^{-\frac{\mu_\beta(\Phi,\Phi)}{2}}$ which enjoys the following properties:
\begin{itemize}
\item[a)] For every $\be>0$, $\om_\beta$ is the unique pure quasifree state which is a KMS state with respect to $v-$translations.
\item[b)] $\om_\beta$ is invariant under the automorphic action of the subgroup $G_{\Im^-}$  generated by rotations and $v-$translations  on $\cW(\Im^-)$.
\item[c)] In the limit $\beta\to\infty$, $\mu_\beta$ converges weakly to the two-point function of \eqref{groundstate}.
\end{itemize}
\end{proposition}
\begin{proof}
As $|\mu(\Phi,\Phi)|\lvertneqq|\mu_\beta(\Phi,\Phi)|$ for all $\beta>0$ and for all $\Phi\in\mathcal{S}(M)$, $\om_\beta$ is a well-defined quasifree (mixed) state. Invariance under the translations generated by $\pa_v$ can be either directly inferred from the explicit form of the two-point function or deduced by repeating the proof of theorem 4.1 in \cite{DMP2}, which additionally entails the uniqueness property. Moreover, invariance under the automorphic action of the rotations and 
$v-$translations 
arises from the same reasoning used in section 4.2 in \cite{DMP2}.

Since the state is quasifree, it is sufficient to verify the KMS condition at the level of two-point functions. Let us consider $\Phi,\Phi'\in\cS(\Im^-)$, and let us introduce the two functions
$$
F(t)\doteq \mu_\beta(\Phi,\alpha_t(\Phi^\prime)) \,,\qquad G(t)\doteq\mu_\beta(\alpha_t(\Phi),\Phi^\prime)\,,
$$
where $\alpha_t(\Phi)(v,\theta,\varphi)=\Phi(v-t,\theta,\varphi)$.
A direct computation shows that both $F(t)$ and $G(t)$ can be seen as the Fourier-Plancherel transform of suitable square-integrable functions $\widehat{F}(E)$ and $\widehat{G}(E)$ which decay faster then any inverse power of $E$. Thus they are continuous functions and, on account of the explicit form of $\mu_B$ \eqref{KMSboundary}, they are related as 
$$
\widehat{F}(E)=e^{\beta E}\widehat{G}(E)\,.
$$
This in turn entails $F(t+i\beta)=G(t)$ and, hence, the validity of the KMS condition as shown in section 5.3 of \cite{Bratteli2}.
\end{proof}

\noindent As discussed in proposition \ref{scalarbtb}, each boundary state induces a bulk counterpart, particularly, we can define
$$
\om^M_\beta\doteq \om_\beta\circ \imath_S\;,
$$
where $\imath_S$ is the injective $*$-homomorphism introduced in proposition \ref{scalarbtb}. As already anticipated, we shall now show that $\om^M_\beta$ fulfils the Hadamard condition as formulated in definition 3.3 of \cite{DHP}, which is valid both for scalar \cite{Radzikowski} and for Dirac fields \cite{Hollands, Kratzert, SahlmannVerch}.

\begin{theorem}\label{Hadamard}
The state $\om^M_\beta:\cW(M)\to\bC$ is a Hadamard state, that is, the integral kernel of its two-point function $$\om_{\be,2}(f,g)\doteq \mu_\be(\Phi_f,\Phi_g)$$ is a distribution in $\mD'(M\times M)$ which enjoys the Hadamard property.
\end{theorem}
\begin{proof}

Per direct inspection, one can realise that the limit $\beta\to\infty$ of \eqref{KMSboundary} is the two-point function of the state \eqref{groundstate} which is of Hadamard form as already proven in \cite{DMP3, Nicola}\footnote{The arguments used in those papers are based on an early proof presented in \cite{Mo08}.}. Thus, in order to conclude the proof, it is sufficient to show that
$$
\De_{\be}(f,h)\doteq\om_{\be,2}(f,h)-\om_{\infty,2}(f,h),
$$
is an element of $\mD'(M\times M)$ with smooth integral kernel. Aiming to prove the continuity of $\De_\be$ first, we set $\Phi_f\doteq\Ga_S(E(f))$ and $\Phi_h\doteq\Ga_S(E(h))\in\cS(\Im^-)$, and directly compute
$$
\De_{\be}(f,h)=\int\limits_{-\infty}^\infty \frac{2 |k|}{e^{\beta |k|}-1}  \overline{\widehat\Phi_f(k,\theta,\varphi)}
\widehat\Phi_h(k,\theta,\varphi)
  \;dk\, d\bS^2\;\doteq \gD_\be(\Phi_f,\Phi_g)\,,
$$
where $\gD_\be=\mu_\be-\mu_\infty$ is the difference of the relevant two-point functions of the boundary states. Since $\frac{2 |k|}{e^{\beta |k|}-1}$ is bounded, $|\De_\beta(f,h)|$ is dominated by the the $L^2$-norms of $\Phi_f$ and $\Phi_h$. Using the continuity of the causal propagator, seen as a map from $\cD(M)$ to $\cE(M)$, the regularity property of the modes $T_k$, the continuity of the Fourier transform in $L^2$ and, finally, the definition of $\Ga_S$, one obtains that
$$
\| k^n \widehat\Phi_f \|_{L^2} \leq C(K,n) \sum_{j=0}^{q(n)} \| D^j f \|_{L^\infty}\,,
$$
where $\widehat{\cdot}$ is the Fourier-Plancherel transform on $\Im^-$, $C(K,n)$ is a constant which depends both on the compact set $K\subset M$ which contains the support of $f$ and on the order $n$, and $\| D^j f \|_{L^\infty}$ has to be understood as the ${L^\infty}$-norm of the $j$-th partial directional derivative evaluated on the $j$ directions which give the maximal result in an arbitrary but fixed local coordinate system. This observation together with the previous ones yields the sought continuity, {\it viz.}
$$
|\De_\beta(f,h)| \leq C_K \sum_{j=0}^{q} \| D^j f \|_\infty \sum_{i=0}^{q} \| D^i h \|_\infty\,,
$$
where, once more, the constant $C_K$ depends on the compact set $K$ which contains the support of both $f$ and $h$, while $q$ is a fixed constant.

Let us now tackle the problem of proving the smoothness of the integral kernel of $\om_{\beta,2}$ defined as the composition $\gD_\beta \circ (\Ga_S(E)\otimes \Ga_S(E))$. To this avail, we can make use of theorem 8.2.13 in \cite{Hormander} which allows us to control the wave front set of this composite linear functional.  To wit, let us recall that, in $\gD_\beta$, $\frac{2 |k|}{e^{\beta |k|}-1}$ decays faster then any inverse power of $|k|$. Hence, if $(x_1,x_2,k_1,k_2)$ is a point in the wavefront set of $\gD_\beta$, $k_1$ and $k_2$ must have vanishing components along causal directions, {\it i.e.} the $v$-direction, but the restriction of $\Ga_S(E)\otimes \Ga_S(E)$ to $\Im^-\times\Im^-$ has a wave front set which is non-vanishing {\it only} in causal directions. Hence, thanks to the above cited theorem, we can conclude that the wavefront set of $\De_\beta$ is empty or, equivalently, that it is smooth.
\end{proof}

\remark{A more detailed analysis of the scalar case and of the preceding theorem in particular can be found in theorem III.2.2.7 of \cite{Thomas}.}

\subsection{The Dirac Field Case}

The case of a Dirac field can be studied following the general ideas of the previous subsection, that is, we first construct a suitable state for the boundary theory and then prove that the natural bulk counterpart it induces is of Hadamard form.

\begin{proposition}
The single-spinor two-point functions defined as
$$
\gw^\pm_\beta(\Psi_1,\Psi_2^\dagger) \doteq \int\limits_{\bR\times\bS^2}  \frac{1}{1+e^{\mp\beta k}}\,       \widehat{\Psi_1}^*(k,\theta,\varphi) \widehat{\Psi_2}(k,\theta,\varphi)\,  dk\,d\bS^2
$$
are distributions on $\gS(\Im^-)\otimes \gS(\Im^-)$. Via the double-spinor two-point function
$$
\gw_\beta(\Psi_{1,1}\oplus \Psi_{1,2}^\dagger,\Psi_{2,1}\oplus \Psi_{2,2}^\dagger)
\doteq \gw^-_\beta(\Psi_{1,1},\Psi_{2,2}^\dagger) + \gw^+_\beta(\Psi_{2,1},\Psi_{1,2}^\dagger)
$$
they unambiguously determine a quasifree state $\Lambda_\beta:\cB(\Im^-)\to\bC$ which is invariant under the automorphic action of the full boundary symmetry group $G_{\Im^-}$. Furthermore, $\La_\beta$ fulfils the KMS condition at inverse temperature $\beta$ with respect to the $v-$translations on $\Im^-$.
\end{proposition}

\begin{proof}
The two-point functions $\gw^\pm_\beta$ are linear and positive functionals and, hence, they induce a unique quasifree algebraic state by well-known results (see chapter 5 in \cite{Bratteli2}). To verify the KMS condition and since $\gw^\pm$ are manifestly invariant under $v$-translations, we can basically repeat the proof for the scalar case: Let us set $\Psi_i\doteq \Psi_{i,1}\oplus \Psi_{i,2}^\dagger$, $i\in\{1,2\}$ and let us introduce
$$
\gF(t)\doteq \gw_\beta(\Psi_1,\alpha_t (\Psi_2)) \,,\qquad
\gG(t)\doteq \gw_\beta(\alpha_t(\Psi_1),\Psi_2)\,,
$$
where $\alpha_t(\Psi_i)(v,\theta,\varphi)\doteq \Psi_i(v-t,\theta,\varphi)$. Both $\gF(t)$ and $\gG(t)$ are bounded and continuous functions which can be seen as the Fourier-Plancherel transforms of two square-integrable functions $\widehat\gF(E)$ and $\widehat \gF(E)$ that satisfy the condition
$$
\widehat \gF(E) = e^{\beta E}\widehat \gG(E)\;,
$$
which is a restatement of the KMS condition. Finally, invariance under the full automorphic action of $G_{\Im^-} $ follows as in the scalar case from the results of \cite{DMP2}, since $\gw^\pm_\beta$ are defined via invariant integrals.
\end{proof}

\remark{Although we do not prove it explicitly, we presume that one can show the just found KMS states on $\cB(\Im^-)$ to be unique and weakly converging to the unique (and pure) ground state with respect to $v$-translations in the limit $\be\to\infty$.}

\noindent As in the scalar case, we can use proposition \ref{Dirachom} to define an algebraic state $\La_\beta:\cB(M)\to\bC$ induced by the boundary counterpart. Yet, before analysing the outcome, we have to introduce a further feature of a Dirac field theory, namely, the concept of charge conjugation: this operation is implemented by the $*$-automorphism $\alpha_c:\cB(M)\to\cB(M)$ defined as $\alpha_c(B(f))\doteq B(f^c)$ for all $f=f_1\oplus f_2$. Here, $f^c\doteq C^{-1}\overline{f}_1\oplus \overline{f}_2C$ where $C\in SL(4,\bC)$ is the matrix which fulfils $\overline{C}C=\bI_4$ and $\overline{\gamma}_a=C\gamma_a C^{-1}$ for all $a\in\{0,1,2,3\}$. Accordingly, it turns out that a quasifree state $\Lambda:\cB(M)\to\bC$ is charge conjugation-invariant, {\it i.e.} $\La\circ \alp_c=\La$, if the associated two-point functions $\gw^\pm$ fulfil  $\gw^\pm(g^{\dagger c},f^{\dagger c})=\gw^\mp(f,g)$ for all $f\in\mD(DM)$ and for all $g\in\mD(D^*M)$. With this in mind, we proceed to prove the main result of this section.

\begin{theorem}
The state $\La^M_\beta:\cB(M)\to\bC$ defined as
\beq\label{pullbackstateDirac}
\La^M_\beta \doteq  \La_\beta \circ \imath_D\;.
\eeq
possesses the following properties.
\begin{itemize}
\item[a)] The state $\La^M_\beta$ is quasifree and charge conjugation-invariant.
\item[b)] The two-point function $\gw_{M,\beta}$ of $\La^M_\beta$ is an element of $\mD'(DM\boxtimes D^*M)$ of Hadamard form.
\end{itemize}
\end{theorem}
\begin{proof}
Let us focus on a). As $\La_\beta$ is already quasifree, $\La^M_\beta$ is manifestly quasifree as well. The charge conjugation-invariance can be checked by explicitly computing that the single-spinor two-point function satisfies the necessary identities, see section III.4 in \cite{Thomas} for details.

Let us now focus on b) and show that $\gw^\pm_{M,\beta}\in\mD'(DM\boxtimes D^*M)$ to start with. This result descends from the continuity of the causal propagator $S$ with respect to the appropriate topologies and from the continuity of the Fourier transform with respect to the relevant $L^2$-norms. Since the modes associated to a solution of the Dirac equation display the same regularity as the modes associated to the Klein-Gordon equation conformally coupled to scalar curvature, we can proceed as in the proof of theorem \ref{Hadamard} to obtain
$$
\| k^n \widehat\Psi_f\|_{L^2} \leq \sum_{j=1}^{q(n)} \|D^j f \|_{L^\infty},
$$
where $\Psi_f \doteq  \Ga_D(Sf)$ and the notation $\|D^j f \|_{L^\infty}$ is the same as the one in theorem \ref{Hadamard}. This part of the proof can then be finalised noticing that $\gw^\pm_{\beta}(\Psi_{f_1},\Psi_{f_2})$ are controlled by the $L^2$ norms $$\|\Psi_{f_i}\|_{L^2}=\int\limits_{\bR\times \bS^2}\widehat \Psi_{f_i}^*(h,\th,\phi)\widehat \Psi_{f_i}(h,\th,\phi)\;dv\,d\bS^2\,,\quad i\in\{1,2\}\,.$$

We are left with the verification of the Hadamard property.  To this end, thanks to the results in \cite{Hollands, Kratzert}, it is sufficient to ensure that the wave front sets of the distributions $\gw_{M,\be}^\pm$ fulfil the requirements of definition IV.1 in \cite{Hollands}. Furthermore, the charge conjugation-invariance of $\La^M_\beta$ implies that $\gw_{M,\beta}^+(f,h)^c = \gw_{M,\beta}^-(h^{\dagger c},f^{\dagger c})$. Hence, it suffices to analyse the wavefront set of $\gw^-$ and
we have to prove that
$$
WF(\gw^-_{M,\beta}) = \left\{(x,y,k_x,k_y)\in  T^*M^2\setminus\{0\}\, | \;(x,k_x)\sim (y,-k_y)  , \, k_x  \triangleright 0  \right\}\,.
$$
As $\gw^-_{M,\beta}=\gw^-_{\beta}\circ(\Ga_DS \otimes \Ga_DS)$, where $S$ is the causal propagator and $\Ga_D$ the projection of solutions to the horizon, this equality can be proven in the same way as the counterpart for the scalar field presented in \cite{DMP3, Nicola}. The wave front set of this composition can be estimated along the lines given in the proof of theorem 4.2 in \cite{Moretti06} where theorem 8.2.13 in \cite{Hormander} is shown to be applicable.
\end{proof}

\subsection{The Bulk Remnant of the Boundary KMS Condition}

Our next goal is to discuss further physical properties of the states $\om^M_\beta$ and $\La^M_\beta$, particularly, we would like to analyse the bulk remnant of the exact KMS condition on the boundary. As a starting point, we write the two-point functions of the states introduced in the preceding subsections in a more manageable form using the bulk mode-decompositions of the scalar and the Dirac field quantities. In the former case, the specific form of the causal propagators written at the end of section \ref{classicalscalar} entails that the two-point function of $\om^M_\beta$ can be written as
$$
\om_{2,\beta}(\tau_x,\vec{x},\tau_y,\vec{y})
=
\frac{1}{(2\pi)^3 a(\tau_x)a(\tau_y)}
\int\limits_{\bR^3}
\at  \frac{\overline{T_k}(\tau_x) T_k(\tau_y)}{1-e^{-\beta k}}   +
\frac{{T_k}(\tau_x) \overline{T_k}(\tau_y)}{e^{\beta k}-1}              \ct  e^{-i\vec{k}{(\vec{x}-\vec{y})}}   \;d^3k\;,
$$
where $T_k(\tau)$ are the modes discussed in and after proposition \ref{KGmodes}.
In the Dirac case, we can make use of the analysis in section \ref{classicalpart} to write the single-spinor two-point function defining  $\La^M_\beta$ as
$$
{\gw^+}_{M,\beta}(\tau_x,\vec{x},\tau_y,\vec{y})=
\int\limits_{\bR^3} \sum_{l\in\{3,4\}}
 \frac{\psi_{\vec{k},l}(\tau_x,\vec{x}) \psi^\dagger_{\vec{k},l}(\tau_y,\vec{y})}{e^{-\beta k}+1}
+
\sum_{l\in\{1,2\}}
\frac{\psi_{\vec{k},l}(\tau_x,\vec{x}) \psi^\dagger_{\vec{k},l}(\tau_y,\vec{y})}{e^{\beta k}+1}\;
d^3k,
$$
where $\psi_{\vec{k},l}$ are the modes introduced in proposition \ref{Diracmod} and appendix \ref{modeanalysis}.

Notice that, although in the massless case a Bose-Einstein or a Fermi-Dirac factor at inverse temperature $\beta$ are present
in the integral kernel of the two-point functions, the bulk spacetimes are not invariant under time translations. Hence, the induced bulk states of the massless theories do not satisfy any KMS condition. However, a KMS-like condition can be found exploiting \eqref{metric2} and the conformal invariance of the considered scalar and Dirac theories. Particularly, we can relate each solution of the relevant equations of motion in a FRW spacetime with flat spatial sections to a counterpart in Minkowski spacetime by a rescaling with a suitable power of the scale factor $a(\tau)$. Consequently, the induced states in the FRW spacetime can now be read as conformally rescaled genuine KMS states at inverse temperature $\beta$ with respect to the Minkowskian time translations. This allows to recover a natural thermal interpretation.

In the case of massive fields, the situation is more involved and a nice and simple analysis as the one above is not possible. Yet, it would be nice if the induced bulk states still enjoyed some properties which could be seen as the remnant of the exact KMS on the boundary. It seems natural to approach the issue by estimating how much the regularised two-point functions deviate from those of the conformally invariant case. We shall analyse only the case of a scalar field, and will briefly comment on the Dirac case afterwards. Hence, we are concerned with the difference
$$
(\om_{\beta,2} - \om_{\infty,2})(\tau_x,\vec{x},\tau_y,\vec{y}) = \frac{1}{(2\pi)^3 a(\tau_x)a(\tau_y)}
\int\limits_{\bR^3}
\at
\frac{\overline{T_k}(\tau_x) T_k(\tau_y)+{T_k}(\tau_x) \overline{T_k}(\tau_y)}{e^{\beta k}-1}              \ct  e^{-i\vec{k}{(\vec{x}-\vec{y})}}   \;d^3k\,,
$$
where the dependence on the mass affects only the form of the modes $T_k$. The fact that these are constructed as a convergent perturbation series which fulfils the initial condition of being asymptotically of the form $\sqrt{2 k}^{\,-1}e^{ik\tau}$, {\it i.e.}, resembling massless positive frequency modes, shall allow us to compare massive and massless theories. For the sake of simplicity, we discuss such a relation only for a specific observable, the Wick-regularised squared scalar field $\wick{\phi^2}$. In the massless and, hence, conformally invariant case we obtain
$$
\at \om^M_\beta - \om^M_\infty\ct (\phi^2(\tau_x,\vec{x})) = \frac{1}{12a^2(\tau_x) \beta^2}\;.
$$
In the massive case, in order to compute the expectation value of the same observable, it suffices to notice that the perturbative construction and initial conditions on $T_k$ imply
$$
\left|\at |T_k(\tau)|^2 -\frac{1}{2k} \ct \right|  \leq \frac{1}{2k} \at  e^{m^2 a(\tau)^2}  -1 \ct\,.
$$
This inequality can be used to control the  deviation from thermal equilibrium, as it can be directly seen from the analysis of the expectation value of $\wick{\phi^2}$, {\it viz.}
$$
\left| a^2\at \om^M_\beta - \om^M_\infty\ct (\phi^2)
\right|
\leq
\frac{1}{12\beta^2}    + \frac{1}{12\beta^2} \at e^{m^2a(\tau)^2}-1\ct .
$$
Although we have only discussed the influence of the mass on a rather special observable, the outcome of the same procedure in other cases such as, for example, the balanced derivatives introduced in \cite{BOR} for the flat case and generalised to curved spacetimes in \cite{Schlemmer}:
$$
\Theta({\bf n},\beta,x)\doteq
\lim_{\xi\to0}(\nabla_{n_1}\dots \nabla_{n_{|{\bf n}|}} \phi   )
(\omega_\beta -\omega_\infty)(\tau+\xi,\vec{x},\tau-\xi,\vec{x}),
$$
can be computed in a similar way. In the massless case, the expectation values of these observables turn out to be fixed and simple functions of the temperature. However, in the case of massive fields, the fact that the thermal nature of these states holds only in an approximated way  manifests itself as a modification of the relations existing between different $\Theta$ -- so-called transport equations -- due to the appearance of central terms, {\it i.e.} sources for the transport equations, depending on both the mass, $a(\tau)$, and its derivatives.

A similar analysis as the one above is certainly possible in the case of Dirac fields. However, the observable which seems to be the natural one in this case, the Wick-square $\wick{\psi^\dagger\psi}$ regularised with the suitable single-spinor two-point function of $\La^M_\infty$, is, in contrast to its scalar counterpart, not a good thermometer since its expectation value in the massless case vanishes.

\subsection{Relaxing the KMS Condition on the Boundary}

In the previous subsection we have found that, in the massive case, the induced states, which fulfil an exact KMS condition on $\Im^-$, have an approximate thermal interpretation in the bulk which is quantitatively controlled by $am$. Hence, one can interpret these states as states which are in perfect thermal equilibrium at $a=0$. From this point of view, a natural generalisation of the states constructed above are states which are in thermal equilibrium ``at an instant of time'' where $a=a_0>0$. To wit, we can define a new set of states $\om^M_{\be,a_0}$ and $\La^M_{\be,a_0}$ as the (charge conjugation-invariant) quasifree states induced by the two-point functions

$$
\om_{2,\beta,a_0}(\tau_x,\vec{x},\tau_y,\vec{y})
=
\frac{1}{(2\pi)^3 a(\tau_x)a(\tau_y)}
\int\limits_{\bR^3}
\at  \frac{\overline{T_k}(\tau_x) T_k(\tau_y)}{1-e^{-\beta \sqrt{k^2+a_0^2m^2}}}   +
\frac{{T_k}(\tau_x) \overline{T_k}(\tau_y)}{e^{\beta \sqrt{k^2+a_0^2m^2}}-1}              \ct  e^{-i\vec{k}{(\vec{x}-\vec{y})}}   \;d^3k\;,
$$and
$$
{\gw^+}_{M,\beta,a_0}(\tau_x,\vec{x},\tau_y,\vec{y})=
\int\limits_{\bR^3} \sum_{l\in\{3,4\}}
 \frac{\psi_{\vec{k},l}(\tau_x,\vec{x}) \psi^\dagger_{\vec{k},l}(\tau_y,\vec{y})}{e^{-\beta \sqrt{k^2+a_0^2m^2}}+1}
+
\sum_{l\in\{1,2\}}
\frac{\psi_{\vec{k},l}(\tau_x,\vec{x}) \psi^\dagger_{\vec{k},l}(\tau_y,\vec{y})}{e^{\beta \sqrt{k^2+a_0^2m^2}}+1}\;
d^3k\,.
$$
These states are still Hadamard as one can prove by trivial modifications of the relevant proofs for the case $a_0=0$. However, the corresponding states on the boundary algebra certainly fail to fulfil an exact KMS condition. Nevertheless, from a physical point of view, one can expect that the states with $a_0>0$ are better suited to describe realistic thermal situations in the early universe than the states with $a_0=0$.

\subsection{Approximated Thermodynamics}
As last point, we briefly sketch how approximated thermodynamic relations hold for the above introduced states, both for $a_0=0$ and for $a_0>0$. In order to support the interpretation of these states as approximated equilibrium ones, we analyse whether and to which extent the usual laws enjoyed by any theory with a sensible notion of thermodynamics hold for certain expectation values. For the sake of simplicity, we consider the scalar field states $\om^M_{\be,a_0}$.
Furthermore, in order to avoid potential complications due to well-known anomalies as the trace anomaly, we shall consider the difference $\om^M_{\beta,a_0}-\om^M_\infty$ and, hence, define
$$\langle A \rangle_{\beta,a_0}=\om^M_{\beta,a_0}(A)-\om^M_\infty(A)\,.$$

As we wish to formulate thermodynamic laws with respect to the evolution generated by $e\doteq\pa_t$ where $t$ is the cosmological time, we introduce, following \cite{Dixon78, BOR}, the vector representing the effective inverse temperature as
$$
{\beta_c}^\mu(x) \doteq \beta_c(x) e^\mu ,\qquad  \beta_c(x)\doteq\frac{12}{\langle\phi^2(x)\rangle_{\beta,a_0}}\,.
$$
Since both $\om^M_{\beta,a_0}$ and  $\om^M_\infty$ are of Hadamard form, $\langle\phi^2(x)\rangle_{\beta,a_0}$ is a well-defined smooth function on $M$. Moreover, this definition allows us to decompose the expectation values of the stress-energy tensor $T_{\mu\nu}$ evaluated once more with respect to $\om^M_{\beta,a_0}-\om^M_\infty$ as
$$
\langle{T_{\mu\nu}}\rangle_{\beta,a_0}\doteq Q e_\mu e_\nu + P g_{\mu\nu}\;,
$$
where $Q$ represents the heat and $P$ the pressure. According to \cite{Dixon78}, we can additionally introduce the vector $S^\mu$ representing the thermodynamic entropy as
$$
S(\beta_c)^\mu \doteq  Q(\beta_c) \beta_c  e^\mu\,.
$$
With these definitions for $\beta_c$, $Q$, $P$ and $S$ at hand, we can discuss potential thermodynamic relations among them. In fact, one can show that the following laws hold near $a=0$ and/or for small $m$. 

\vskip .4cm

\noindent {\em 0th law:} $\beta_c$ is constant on the Cauchy surfaces at fixed cosmological time.

\vskip .4cm

\noindent {\em 1st law:} $\langle T_{\mu\nu}\rangle_{\be,a_0}$ is covariantly conserved.

\vskip .4cm

\noindent {\em 2nd law:} The entropy production is positive, that is, 
    $$ \nabla_\mu S^\mu  = -\frac{4}{3} \langle T^{\mu\nu}-\frac{1}{4}T g^{\mu\nu}      \rangle_{\beta,a_0}  \nabla_{(\mu}{\beta_c}_{\nu)}
    +\frac{1}{3} {\beta_c}^\mu\pa_\mu \langle T \rangle_{\beta,a_0}\ge0$$

\vskip .4cm

\noindent {\em 3rd law:} The minimum of the entropy function is reached in the limit $\beta_c\to \infty$.

\vskip .4cm

\noindent
The most difficult relation to prove is the second one. In the massless case, the entropy production is vanishing because $\beta_c$ turns out to be exactly a conformal Killing vector and $\langle T\rangle_{\beta,a_0}$ vanishes. In order to address the massive case with $a_0=0$, we can make an expansion of the relevant quantities in powers of $m$ and finally restrict our attention to the region near  $\Im^-$. Performing such an approximation, we observe that $\beta_c^m$ computed for a mass $m$, differs from the massless one $\beta_c^0$  via $$\beta_c^m-\beta_c^0\doteq(- c_1  m^2 a^3 )+ O(m^3)\,,$$ where $c_1$ is a positive constant.
Furthermore, the first order in $m^2$ of $\langle T \rangle_{\beta,a_0}$ is equal to $ -c_2 \frac{m^2}{a^2}$ where $c_2$ is another positive  constant. Finally, we notice that $|\langle T^{\mu\nu}-\frac{1}{4}T g^{\mu\nu}      \rangle_{\beta,a_0}| \leq C a^{-4}$, and that $\langle T^{\mu\nu}-\frac{1}{4}T g^{\mu\nu}      \rangle_{\beta,a_0}  \nabla_{(\mu}{\beta^0_c}_{\nu)}$  vanishes because $\beta^0_c$ is a conformal Killing vector. With these observations in mind we can compute $a^2 \nabla S$ up to $O(m^3)$ terms and notice that it tends to a positive constant towards $\Im^-$. Hence since $a^2$ is positive and since  $a^2 \nabla S$ is smooth in $M$, $\nabla S$ must be positive near $\Im^+$. In the case $a_0>0$, we have only been able to verify the positivity of the entropy production by numerical methods, as the occurring integrals can not be evaluated analytically any more.

\section{Conclusions}
We have studied the classical and the quantum behaviour both of a real, massive, conformally coupled scalar field
and of a massive free Dirac field in a large class of Friedmann-Robertson-Walker spacetimes $(M,g_{FRW})$ with flat spatial sections. These spacetimes are singled out by criteria which have both a physical and a mathematical origin. Particularly, we have constrained the dependence of the scale factor $a(t)$ on the cosmological time $t$ to be such that the range of the conformal time $\tau$ is unbounded on the negative real axis. Moreover, we have demanded that the dependence of $a(\tau)$ on $\tau$ is of the form \eqref{scalef}, which entails that the corresponding Universe underwent a phase of either exponential or power-law expansion in its early stages. From a geometrical point of view, our requirement on $a(t)$ entails that the underlying manifold can be endowed with a boundary $\Im^-$ which is either a cosmological horizon or a Big Bang hypersurface, and mathematically a null differentiable manifold of codimension $1$. On $\Im^-$ it is possible to intrinsically construct both a scalar and a Dirac field theory whose associated algebras of quantum observables are sufficiently big to contain the bulk counterpart via an injective $*$-homomorphism. The motivation to introduce this injective homomorphism lies in the particular topological structure of $\Im^-$, which allows to define an exact notion of a quasifree KMS state $\omega_\beta$ for the boundary theory with respect to the rigid translations along the complete null geodesics generating $\Im^-$. Furthermore, the above-mentioned injection allows to pull-back $\omega_\beta$ to the bulk, yielding a quasifree state which, both in the scalar and in the Dirac case, has various physically and mathematically relevant properties. To wit, such a state always fulfils the Hadamard condition and, when the mass vanishes, also the KMS one.

From a physical point of view, our results have two immediate possible applications: on the one hand, we can now extend the analysis of \cite{DFP} and look for stable, homogeneous, and isotropic solutions of the semiclassical Einstein equations driven by quantum fields of higher spin. This has been done in \cite{Thomas, dhmp}, and interesting results have been found. On the other hand, the states we have introduced are both Hadamard and possess an approximate thermodynamic interpretation also for massive fields. Hence, one can expect that they provide a natural starting point to analyse the cosmic neutrino background on a genuine quantum field theoretic level. It can be expected that the deviations from thermal equilibrium give rise to modifications of the standard picture, in contrast to the cosmic microwave background, which relies on a well-defined thermodynamic interpretation, as photons are massless. We hope to return to this issue in the near future.

\vspace*{9mm}
{\noindent\bf Acknowledgements}\\
C.D. gratefully acknowledges financial support from the Junior Fellowship Programme of the Erwin Schr\"odinger Institute and from the German Research Foundation DFG through the Emmy Noether Fellowship WO 1447/1-1. The work of T-P.H. is supported by the research clusters SFB676 and LEXI ``Connecting Particles with the Cosmos''. N.P. is supported in part by the ERC Advanced Grant 227458 OACFT ``Operator Algebras and Conformal Field Theory". C.D. and N.P. would like to thank the Erwin Schr\"odinger Institute in Vienna and the organisers of the workshop ``Quantum Field Theory on Curved Space-times and Curved Target Spaces'' for their warm hospitality.

\appendix
\section{On the Definition of Dirac Fields}\label{appdirac}
The role of this section is to provide a short introduction to the mathematical structures which are needed to define the notion of Dirac fields on a four-dimensional globally hyperbolic spacetime. Our goal is not to provide an exhaustive analysis but rather to offer to a potential reader the chance to understand the content of this paper without necessarily resorting to other long introductory manuscripts. Hence, this appendix shall summarise what has already been presented in \cite{Sanders, DHP, Thomas}.

In this paper we only consider Friedmann-Robertson-Walkers spacetimes $M$, which are 
four-dimensional, Hausdorff, connected, smooth, simply connected, oriented and time-oriented globally hyperbolic manifolds endowed with a Lorentzian metric, whose signature is chosen as $(-,+,+,+)$. Therefore, we can introduce two important ingredients:
\begin{itemize}
\item The {\bf spin group} $Spin(1,3)$ as the double cover of
$SO(1,3)$, {\it i.e.}, there exists the following short exact sequence of Lie group homomorphisms:
$$\left\{e\right\}\longrightarrow\mathbb{Z}_2\longrightarrow Spin(1,3)\longrightarrow SO(1,3)\longrightarrow
\left\{e\right\},$$
where $\left\{e\right\}$ stands for the trivial group, whereas $\mathbb{Z}_2\doteq\left\{\pm 1\right\}$ is the
cyclic group of order $2$.

\item The {\bf frame bundle} associated to the tangent bundle $TM$, that is, the principal bundle
$$FM=FM[SO_0(3,1), \pi',  M],$$
where $SO_0(3,1)$ is the component of the Lorentz group connected to the identity. Notice that $FM$ is build from the disjoint union $\bigsqcup_xF_xM$, where $F_xM$ is identified with the typical fibre $SO_0(3,1)$ and where $\pi':FM\to M$ is the projection map.
\end{itemize}

\noindent We are now in the position to introduce the main geometric structure at the heart of the construction and of the analysis of Dirac (co)spinor fields.

\begin{definition}\label{spinstr}
Given an oriented and time oriented spacetime $M$, a
{\bf spin structure} is the pair $(SM,\rho)$ where $SM\doteq SM[Spin_0(3,1),\widetilde\pi,M]$ is a principal  fibre bundle
over $M$ with the identity component of the spin group as typical fibre. Moreover, $\rho$ is a smooth
equivariant bundle morphism from $SM$ to $FM$, that is, the following two conditions hold:
\begin{enumerate}
\item $\rho$ is base point preserving, such that
$$\pi'\circ\rho =\widetilde\pi,$$
\item $\rho$ must be equivariant, i.e., calling $R_{\widetilde\Lambda}$ and $R_\Lambda$ the natural right
actions of $Spin_0(3,1)$ on $SM$ and of $SO_0(3,1)$ on $FM$ respectively, we require that
$$\rho\circ R_{\widetilde\Lambda}=R_{\Lambda}\circ\rho,\qquad\forall\Lambda\in SO_0(3,1),$$
where $\Lambda=\Pi(\widetilde\Lambda)$, being $\Pi$ is the surjective covering from $Spin(3,1)$ to $SO(3,1)$.
\end{enumerate}
\end{definition}

\noindent For the scenario we are interested in, the fact that $M$ is four-dimensional and globally hyperbolic not only assures that a spin structure exists, but this structure is also unique up to equivalence, since $M$ is a simply connected spacetime \cite{Geroch, Geroch2} ({\it c.f.}, \cite{Dimock} for the definition of equivalence in this context). The kinematically allowed configurations can now be readily defined out of the objects at hand.

\begin{definition}We call {\bf Dirac bundle} of a four dimensional Lorentzian spacetime M the $\bC^4$-bundle $DM\doteq SM\times_T\bC^4$ associated to $SM$ with respect to the
representation $T\doteq D^{\left(\frac{1}{2},0\right)}\oplus D^{\left(0,\frac{1}{2}\right)} $ of $SL(2,\bC)\sim Spin_0(3,1)$. This is the set of equivalence classes $[(p,
z)]$, where $p\in SM$, $z\in\bC^4$ and equivalence is defined out of the relation
$$(p_1,z_1)\sim(p_2,z_2),$$
if and only if there exists an element $A$ of $SL(2,\bC)$ such that $R_A(p_1)=p_2$ and $T(A^{-1})z_1=z_2$. The global structure of
$DM$ is that of a fibre bundle over $M$ with typical fibre $\bC^4$, and the projection map $\pi_D$ is traded
from the one of $SM$, namely, $\forall\,[(p,z)]\in DM$, it holds
$$\pi_D[(p,z)]\doteq\widetilde\pi(p).$$
Furthermore, if we endow $\bC^4$ with the standard non-degenerate internal product, we can
construct the {\bf dual Dirac bundle} $D^*M$ as the $\bC^{4*}$-bundle associated to $SM$ requiring that
$(p_1,z_1^*)$ and $(R_A(p_1),z_1^*T(A))$ are equivalent, where $^*$ denotes the adjoint with respect to the
inner product on $\bC^4$ and elements of $\bC^{4*}$ are understood as row vectors. Consequently, the dual
pairing of $\bC^4$ and $\bC^{4*}$ extends in a well-defined way to a fibrewise dual pairing of $DM$ and
$D^*M$.
\end{definition}
\noindent According to this definition, one can introduce the following notions.
\begin{itemize}
\item A {\bf Dirac spinor} is a smooth global section of the Dirac bundle, {\it i.e.},
$\psi\in\cE(DM)$. Since $DM$ is trivial due to $M$ being four-dimensional and simply connected, $\psi$ is (diffeomorphic to) a vector-valued function $\psi:M\to\bC^4$.
\item We call {\bf Dirac cospinor} a smooth
global section of the dual Dirac bundle, namely, $\psi'\in\cE(D^*M)$. On a FRW spacetime, $\psi':M\to\bC^{4*}$.
\end{itemize}

\section{Perturbative Construction and Analysis of the Classical Solutions}\label{modeanalysis}

In order to have full control on the behaviour of the classical solutions of both the scalar and of the Dirac Cauchy problems, we must discuss the particular form of the functions $T_k$ and $\chi_{k}$ which are present in the expansions \eqref{TK} and \eqref{ODE}, respectively. Since the analysis of the scalar case has been already performed in \cite{DMP2,DMP3, Nicola}, the aim of this section will be to only address the Dirac case. If we focus on \eqref{ODE}, then the construction of each $\chi_{k}$ can be performed along the same lines pursued in the scalar case. Moreover, we only discuss the asymptotic de Sitter case specified by $\de=0$ in \eqref{scalef} and briefly mention the modifications necessary to treat the asymptotic power-law case, {\it i.e.} $\de>0$, afterwards.

In other words, we first consider \eqref{ODE} in the de-Sitter Universe and, then, we solve the general case by means of a convergent perturbation series. Hence the starting point is \eqref{ODE} written as
\beq\label{pert}
\left(\frac{d^2}{d\tau^2}+V_0(k,\tau)+V(\tau)\right)\chi_{k}(\tau)=0\,,
\eeq
where $V_0(k,\tau)\doteq k^2+\left(\frac{m^2}{H^2}-i\frac{m}{H}\right)\frac{1}{\tau^2}$ whereas $V(\tau)\doteq
k^2+a(\tau)^2m^2-ia^\prime(\tau)m - V_0(k,\tau)$. Notice, that, according to \eqref{scalef} with $\gamma=-H^{-1}$, the function $V(\tau)$ is either $0$ or at least of order $O(\tau^{-3})$.

Since we would like to treat $V(\tau)$ as a perturbation potential, we obviously start from the unperturbed solution.
If $V(\tau)=0$ we are in a pure de Sitter spacetime and two solutions which satisfy the desired initial conditions are
\beq\label{sol}
\chi^0_{k,1} = A \frac{\sqrt{-\pi \tau}}{2} H_\nu^{(2)}(-k\tau) \; \qquad
\chi^0_{k,2} = \overline{B \frac{\sqrt{-\pi \tau}}{2} H_\nu^{(1)}(-k\tau)},
\eeq
where $H_\nu^{(1)}$ and $H_\nu^{(2)}$ are the first and second Hankel functions with
\beq\label{nu}
\nu\doteq
\frac{1}{2}+i\frac{m}{H} \;.
\eeq
Furthermore, in order to satisfy the desired normalisation \eqref{normalisation} we  have
$$
A= \frac{e^{-i\frac{\pi}{2}\nu -i \frac{\pi}{4} }}{\sqrt{2k}}\;, \qquad
B= \frac{e^{+i\frac{\pi}{2}\nu +i \frac{\pi}{4} }}{\sqrt{2k}}.
$$
Since we have a full control of the solutions of \eqref{ODE} in the cosmological de Sitter spacetime, we can now revert back our attention to \eqref{pert} and we shall simply proceed along the same lines of \cite{DMP2, DMP3} for the real scalar field. Hence we write a formal solution as a Dyson-Duhamel series:
\beq\label{Dysonseries}
\chi_{k}(\tau_0)=\chi_{k,1}^0(\tau_0)+\sum\limits_{n=1}^\infty(-1)^n\int\limits_{-\infty}^{\tau_0} d\tau_1\int\limits_{-\infty}^{\tau_1} d\tau_2...\int\limits_{-\infty}^{\tau_{n-1}} d\tau_n
\prod\limits_{i=1}^{n} S_k(\tau_{i-1},\tau_i)V(\tau_i)\chi^0_{k,1}(\tau_n),
\eeq
where $\tau_0=\tau$ and where $S_k(\tau,\tau')$ is the retarded fundamental solution of \eqref{ODE} in cosmological de Sitter spacetime:
$$
S_k(\tau,\tau')\doteq -i2k(\chi^0_{k,1}(\tau)\overline{\chi^0_{k,2}(\tau')}-\chi^0_{k,1}(\tau')\overline{\chi^0_{k,2}(\tau)}).
$$
Notice that, being antisymmetric, $S_k(\tau,\tau)=0$, while the time derivative is a conserved quantity of the underlying dynamical system and thus we can simply set $\left.\frac{d S_k(\tau,\tau')}{d\tau}\right|_{\tau=\tau'}=1$. Of course \eqref{Dysonseries} would be rather useless if we were not able to prove its convergence and, to this avail, we need to provide suitable estimates for both $\chi^0_{k,j}$ and for the involved integrand.

Let us start from $S_k(\tau,\tau')$ and, since we would like to have an estimate for $\chi_{k}(\tau_0)$ which guaranties square integrability in $k$ and at most polynomial growth for large $k$, we shall give two estimates, for small and large $k$.
Let us start analysing the large $k$ behaviour.
We observe that, being $H^{(2)}_\nu(z)=\overline{H^{(1)}_{\overline{\nu}}(\overline{z})}$, we can make use of the integral representation present in formula 9 in 8.421 of \cite{Grad}, in order to express
both $\chi^0_{k,j}$ in a more manageable form.
We notice that the integral representation can be uniformly estimated in order to give
\beq
|\chi^0_{k,j}|\leq C_{\nu} \frac{1}{k},
\eeq
where $C_\nu$ is some positive constant.
With this result we obtain, by direct inspection that
\beq\label{Slargek}
|S_k(\tau,\tau')| \leq C \frac{1}{k},
\eeq
holds uniformly in $\tau$ and $\tau'$ for some positive constant $C$.
The latter will be suitable for analysing the convergence of the perturbative series for large value of $k$.
We proceed now to derive another estimate in order to control the behaviour of $|\chi_{k}|$ for small values of $k$.
We notice that $S_k(\tau,\tau')$ can be rewritten in terms of the Bessel $J_\nu$ functions in the place of the Hankel one, to obtain that
$$
|S_k(\tau,\tau')| \leq  C_\nu'\sqrt{\tau\tau'}
\left|J_\nu(-k\tau)J_{-\nu}(-k\tau')-J_\nu(-k\tau')J_{-\nu}(-k\tau)\right|,
$$
where $C_\nu$ is again a certain positive constant.
We can now make use of the recursive relations to rewrite $J_{-\nu}(z)$ in terms of the $J_{-\nu+1}$ and its first derivative (formula 1 in 8.472 in \cite{Grad}).
Now both Bessel functions can be estimated analysing their integral representation present in formula 5 in 8.411 in \cite{Grad}.
In that way we obtain the following estimate valid for $|\tau|<|\tau'|$
\beq\label{Ssmallk}
|S_k(\tau,\tau')| \leq  C \at |\tau'| + k\; |\tau'|^2 \ct,
\eeq
where $C$ is some positive constant.

With these estimates in mind, we can now analyse the convergence of the series \eqref{Dysonseries}.
From \eqref{Ssmallk} we obtain that the following uniform estimate holds
\beq\label{smallk}
|\chi_{k}(\tau)-\chi^0_{k,1}(\tau)|
\leq
\at \exp \int\limits_{-\infty}^\tau C \at |\tau_1| + k\; |\tau_1|^2\ct-1\ct
\, |V(\tau_1)|\; d\tau_1 \sup_{\tau_2\in(-\infty,\tau)}|\chi_{k,1}^0(\tau_2)| \;,
\eeq
while, from \eqref{Slargek}, the series \eqref{Dysonseries} can be uniformly estimated by
\beq
|\chi_{k}(\tau)-\chi^0_{k,1}(\tau)|
\leq\label{largek}
\at \exp \int\limits_{-\infty}^\tau \frac{C}{k} -1 \ct
\, |V(\tau_1)|\; d\tau_1 \sup_{\tau_2\in(-\infty,\tau)}|\chi_{k,1}^0(\tau_2)| \;.
\eeq
If $|V(\tau)| \leq C |\tau|^{3+\epsilon}$, both integrals present in the preceding  inequalities are finite.

We proceed now to briefly discuss the approximation we shall use for the time derivative of $\chi_{k}(\tau)$.
Since the employed procedure is the same as the one written above, we simply summarise the result.
We obtain that also the series for the first order derivative of $\chi_{k}$ in the time variable is uniformly bounded by a convergent series, namely
\begin{gather}\label{estimatepachi}
\left|\frac{d}{d\tau} \chi_{k}(\tau) - \frac{d}{d\tau} \chi^0_{k,1}(\tau)   \right| \leq \notag
\\
\leq C \int\limits_{-\infty}^\tau
W_2(k,\tau,\tau_2)
|V(\tau_2)| d\tau_2   \;
\exp \int\limits_{-\infty}^\tau W_1(\tau_1,k)\;|V(\tau_1)|\; d\tau_1 \sup_{\tau_3\in(-\infty,\tau)}|\chi_{k,1}^0(\tau_3)| \;.
\end{gather}
where $W_1(\tau_1,k)$ is a positive real function greater than $|S_k{(\tau,\tau_1)}|$ uniformly in $\tau$.
Hence, according both to \eqref{Ssmallk} and to \eqref{Slargek}, $W_1(\tau_1,k)$ can be either $C(|\tau_1|+k|\tau_1|^2)$ or $k^{-1}$.
Furthermore, $W_2(k,\tau,\tau_2)$ is another positive real function that controls uniformly the derivative of the propagator, that is
$|\frac{\pa}{\pa\tau } S_k{(\tau,\tau_2)}|\leq W_2(k,\tau,\tau_2)$.
Proceeding as above we observe that $W_2$ can be either
$C(1+\frac{1}{|k\tau_2|})$,  which is an appropriate bound for large values of $k$, or
$C(1+ \frac{|\tau_2|}{|\tau|}+{|k\tau_2|}+{|k\tau_2|}^2)$, which is better suited for small $k$.
We summarise the discussion as:

\begin{lemma}\label{modesregularity}
For sufficiently large values of $|\tau|$, the series \eqref{Dysonseries} for both $\chi_{k}$ and its first time derivative $\pa_\tau\chi_{k}$ is uniformly convergent if $|V(\tau)| \leq C|\tau|^{3+\epsilon}$.
Furthermore, for fixed $\tau$, both $\chi_{k}(\tau)$ and $\pa_\tau\chi_{k}(\tau)$ are contained in
$L^2([0,k_1],k^2dk)$ for any $k_1>0$.
For large values of $k$ both $\chi_{k}(\tau)$ and its first time derivative are at most growing polynomially.
\end{lemma}
\begin{proof}
The convergence of the series descends easily from the discussion presented above.
The local square integrability shown by $\chi_{k}$ descends from
\eqref{smallk} while the polynomial boundedness for large $k$ from
\eqref{largek}. The very same properties shown by $\pa_\tau\chi_{k}$ are obtained from \eqref{estimatepachi}.
\end{proof}

As anticipated, we now briefly mention how analogous results can be obtained in the asymptotic power-law case. In this case, the solutions can be constructed as a perturbative series around the solutions of the massless, conformally coupled Klein-Gordon equation. Hence, $V_0(\tau)\equiv0$, whereas $V(\tau)=ia^\prime(\tau) m+a^2(\tau)m^2$. Suitable estimates for $V(\tau)$ follow trivially from \eqref{scalef}, consequently, the convergence and $k$-regularity properties of the perturbation series can be discussed as above. We refer the reader interested in further details to the proof of lemma II.4.1.5 in \cite{Thomas}.

\end{document}